\documentclass[journal]{IEEEtran}
\usepackage{amsmath}
\usepackage{amsthm}
\usepackage{amsfonts}
\usepackage{amssymb}
\usepackage{authblk,multirow,multicol}
\usepackage{graphicx,array,cite,caption,mathrsfs}
\usepackage{epstopdf,placeins}
\usepackage{mathtools}
\usepackage{algorithm,float}
\usepackage{algpseudocode}
\usepackage{algorithmicx}
\usepackage{tikz}
\usetikzlibrary{shapes,arrows}
\tikzstyle{line}=[draw]
\usepackage{bigstrut}

\floatstyle{ruled}
\usepackage[utf8]{inputenc}

\newfloat{algorithm}{tbp}{loa}
\providecommand{\algorithmname}{Algorithm}
\floatname{algorithm}{\protect\algorithmname}
\theoremstyle{remark}
\newtheorem{thm}{Theorem}

\newtheorem{lem}{Lemma}

\newtheorem{defn}{Definition}
\newtheorem{exmp}{Example}

\newtheorem{rem}{Remark}

\ifCLASSINFOpdf
\else
\fi

\algdef{SE}[DOWHILE]{Do}{doWhile}{\algorithmicdo}[1]{\algorithmicwhile\ #1}

\title{Groupcast Index Coding Problem: Joint Extensions}
\author{Chinmayananda Arunachala, and B. Sundar Rajan. \thanks{C. Arunachala and B. S. Rajan are with the Dept.
		of Electrical Communication Engg., Indian Institute of Science, Bengaluru
		560012, KA, India, email: \{chinmayanand,bsrajan\}@iisc.ac.in.}}

\begin{document}
\maketitle
\begin{abstract}
The groupcast index coding problem is the most general version of the classical index coding problem, where any receiver can demand messages that are also demanded by other receivers. Any groupcast index coding problem is described by its \emph{fitting matrix} which contains unknown entries along with $1$'s and $0$'s. The problem of finding an optimal scalar linear code is equivalent to completing this matrix with known entries such that the rank of the resulting matrix is minimized. Any row basis of such a completion gives an optimal \emph{scalar linear} code. An index coding problem is said to be a joint extension of a finite number of index coding problems, if the fitting matrices of these problems are disjoint submatrices of the fitting matrix of  the jointly extended problem. In this paper, a class of joint extensions of any finite number of groupcast index coding problems is identified, where the relation between the fitting matrices of the sub-problems present in the fitting matrix of the jointly extended problem is defined by a base problem. A lower bound on the \emph{minrank} (optimal scalar linear codelength) of the jointly extended problem is given in terms of those of the sub-problems. This lower bound also has a dependence on the base problem and is operationally useful in finding lower bounds of the jointly extended problems when the minranks of all the sub-problems are known. We provide an algorithm to construct scalar linear codes (not optimal in general), for any groupcast problem belonging to the class of jointly extended problems identified in this paper. The algorithm uses scalar linear codes of all the sub-problems and the base problem. We also identify some subclasses, where the constructed codes are scalar linear optimal. 
\end{abstract}

\section{Introduction}
The index coding problem (ICP) introduced in \cite{BK} is a source coding problem with some side-information present at the receivers. The sender broadcasts coded messages leveraging the knowledge of the side-information present at all the receivers, in order to reduce the number of transmissions required for all the receivers to decode their demanded messages. This problem  is also related to topological interference management problem in wireless networks \cite{jafar}. It also has applications in satellite communications where some users  want to exchange their messages using a satellite \cite{SUOH}, and the retransmission phase of downlink networks \cite{LNSG} among many others. In general, the ICP is NP-Hard. Optimal codelengths and optimal codes were given for some special classes of the ICP \cite{SUOH}, \cite{MV}. Many works address the  single unicast ICP (SUICP) where each receiver demands a unique message \cite{LS}, \cite{DSC1}. The most general class of the ICP which subsumes SUICP is the groupcast ICP  where any receiver can demand messages which are also demanded by other receivers.

The groupcast ICP was first studied in \cite{LS} where upper and lower bounds on the optimal codelength were given for any groupcast problem in terms of the optimal codelengths of two related SUICPs. In \cite{DIC}, a directed bipartite graph representation was introduced for the groupcast problem and capacity region was found when only particular coding schemes were allowed. The groupcast problem was represented as a directed hypergraph and bounds on the optimal broadcast rate were given in \cite{ALSW}. In \cite{SKS}, optimal scalar codelengths were obtained for a class of the groupcast problem, where each message is demanded by at most two receivers. The results are obtained based on the optimality of linear coding schemes for a related SUICP. 

Characterisation of the optimal codelengths of SUICPs in terms of those of its sub-problems has been carried out in many works \cite{MV}, \cite{RKMV}-\cite{FYGL}. A lifting construction was presented in \cite{RKMV}, where a special class of SUICPs were obtained from another class of SUICPs. The  optimal scalar linear  codelength of the larger derived SUICP has been shown to be equal to that of the smaller SUICP. Optimal vector linear codes for a class of SUICPs were constructed using optimal scalar linear codes of other basic SUICPs in \cite{MV}. Graph homomorphism between complements of the side-information digraphs of two given SUICPs was used to establish a relation between their optimal codelengths \cite{JS}. Some special classes of rank invariant extensions of any SUICP were presented in \cite{PK}, where the extended problems have the same optimal linear codelength as that of the original SUICP, generalizing the results of \cite{RKMV}. The notion of rank invariant extensions was extended to a class of joint extensions of any finite number of SUICPs in \cite{CBSR}. Two-sender SUICPs with a sub-problem being a joint extension of two SUICPs were solved for optimal scalar linear codelengths using those of the component single-sender SUICPs (sub-problems) \cite{CBSR}. In \cite{FYGL}, capacity region of SUICPs with side-information digraphs being  generalized lexicographic products of side-information digraphs of the component SUICPs was characterized in terms of those of the component SUICPs.

In this paper, we identify a class of joint extensions of a finite number of groupcast ICPs, where the relation between the sub-problems in the jointly extended problem is defined by a base problem. Optimal scalar linear codelength and optimal codes of the jointly extended problem are given in terms of those of the sub-problems and the base problem for a special class of the jointly extended problem introduced in this paper. This result generalizes the class of joint extensions solved in \cite{CBSR}. When the base problem and all the sub-problems are restricted to SUICPs, the class of jointly extended groupcast problems identified in this paper reduces to the class of SUICPs with the  side-information digraphs being generalized lexicographic products of the component side-information digraphs \cite{FYGL}. Viewing any jointly extended problem from a matrix-completion perspective, the class of jointly extended problems solved in this paper extends the notion of generalized lexicographic products where any number of sub-problems can be groupcast ICPs.

The key results of this paper are summarized as follows. 
\begin{itemize}
    \item A class of joint extensions is identified which extends the notion of generalized lexicographic product of side-information digraphs which has been  defined with the sub-problems being SUICPs \cite{FYGL}, to the case where the sub-problems can be groupcast problems. The positions of the fitting matrices of the sub-problems (in that of the extended problem given in this work), are given by the entries of the fitting matrix of another ICP called the base problem.  
	\item  A lower bound on the \emph{minrank} (optimal scalar linear codelength) is given for the class of jointly extended problems identified in this paper, in terms of those of the sub-problems and the base problem.
    \item A code construction (not necessarily optimal) is presented for a sub-class of the class of jointly extended problems, based on codes of the base problem and all the sub-problems. A set of necessary conditions are provided for the constructed codes to be optimal.
    \item  An algorithm to construct a scalar linear code using any given scalar linear codes of the sub-problems and the base problem is given. The constructed code need not be scalar linear optimal, even when all the related codes are scalar linear optimal. This is the first work (to the best of our knowledge) where deterministic/explicit codes are constructed  for a bigger groupcast problem using those of many smaller groupcast sub-problems. A subclass of joint extensions with the constructed codes being optimal scalar linear codes is identified.
\end{itemize}
\par The remainder of the paper is organized as follows. Section II introduces the problem setup and establishes the notations and  definitions used in this paper. Section III contains our initial results on jointly extended problems presented in the first version of this paper. Section IV presents improved results. We provide an algorithm to obtain scalar linear codes for the class of jointly extended problems identified in this paper. Section V identifies a subclass of jointly extended problems, where the algorithm provides optimal scalar linear codes. Section VI concludes the paper with directions for future work.

\section{Problem Formulation and Definitions}

In this section, we establish the notations and definitions used in this paper and formulate a class of the groupcast index coding problem that can be seen as joint extensions of  smaller groupcast problems. 

Matrices and vectors are denoted by bold uppercase and bold lowercase letters respectively. For any positive integer $m$, $[m]  \triangleq \{1,...,m\}$. $\mathbb{F}_q$ denotes the finite field of order $q$. $\mathbb{F}_q^{n \times d}$ denotes the vector space of all $n \times d $ matrices over $\mathbb{F}_q$.

We first define the notion of disjoint submatrices of a given matrix, which is needed to define the class of jointly extended problems dealt in this paper, and provide related notations.
 
\begin{defn}[Disjoint submatrices of a given matrix]
A matrix ${\bf{L}}$ obtained by deleting some of the rows and/or some of the columns of ${\bf{M}}$ is said to be a submatrix of ${\bf{M}}$. The notation ${\bf{L}} \prec {\bf{M}}$ denotes that ${\bf{L}}$ is a submatrix of ${\bf{M}}$. The set containing the indices of columns of ${\bf{M}}$ present in ${\bf{L}}$ is denoted by $col({\bf{L}},{\bf{M}})$. Similarly, the set containing the indices of rows of ${\bf{M}}$ present in ${\bf{L}}$ is denoted by $row({\bf{L}},{\bf{M}})$. Indices of columns and rows are assumed to start from $1$. A set of submatrices of a given matrix are said to be disjoint, if no two of the submatrices have  elements indexed by the same ordered pair in the given matrix.
\end{defn}

The number of rows and columns of any matrix ${\bf{M}}$ are denoted by $\mathcal{R}({\bf{M}})$ and $\mathcal{C}({\bf{M}})$ respectively. The $(i,j)$th entry of matrix ${\bf{M}}$ is denoted as ${\bf{M}}_{i,j}$. The notation $\big[[{\bf{M}}]\big]_{i,j}$ denotes the $(i,j)$th component block matrix of ${\bf{M}}$, where the component block matrices of ${\bf{M}}$ (or equivalently the partition of ${\bf{M}}$  into component block matrices) are predefined by the construction of ${\bf{M}}$ using the same.  ${\bf{M}}_{[\mathcal{R}]}$ denotes the matrix formed by stacking the rows of ${\bf{M}}$ indexed by the elements in the set $\mathcal{R}$ in the ascending order of indices such that the row with the least row index forms the first row of ${\bf{M}}_{[\mathcal{R}]}$. For any matrix ${\bf{M}}$ over $\mathbb{F}_q$, the rank of ${\bf{M}}$ over $\mathbb{F}_q$ is denoted by $rk_{q}({\bf{M}})$. $\langle {\bf{M}} \rangle$ denotes the row space of ${\bf{M}}$. The transpose of ${\bf{M}}$ is denoted by ${\bf{M}}^{T}$.

We now define an upper-triangulable matrix which is frequently used in this paper. 
\begin{defn}[Upper-triangulable matrices] 
A permutation matrix ${\bf{P}}$ is a square matrix that has exactly one $1$ in each row and each column and $0$'s elsewhere. Any $p \times p$ permutation matrix ${\bf{P}}$ represents a permutation of $p$ elements. For a $p \times p$ matrix ${\bf{M}}_x$ containing unknown elements denoted by $x$ along with some known elements, ${\bf{P}}{\bf{M}}_x$ denotes the matrix obtained by applying the permutation described by ${\bf{P}}$ on the rows of ${\bf{M}}_x$. Similarly, ${\bf{M}}_x{\bf{P}}$ denotes the matrix obtained by applying the permutation described by ${\bf{P}}$ on the columns of ${\bf{M}}_x$. A $p \times p$ square matrix ${\bf{M}}$ is said to be upper-triangulable if there exists two $p \times p$ permutation matrices ${\bf{P}}$ and ${\bf{Q}}$ such that ${\bf{P}}{\bf{M}}{\bf{Q}}$ is an upper-triangular matrix. A matrix constructed using block matrices is called block upper-triangular, if the matrix obtained by replacing each block matrix by a scalar is upper-triangular. The block matrices can also be rectangular matrices. A $p \times p$ matrix ${\bf{M}}_x$ (with some unknown elements denoted by $x$) is said to be upper-triangulable if there exists two $p \times p$ permutation matrices ${\bf{P}}$ and ${\bf{Q}}$ such that ${\bf{P}}{\bf{M}}_x{\bf{Q}}$ is an upper-triangular matrix with all the diagonal entries being equal to $1$. The set of all  $p \times p$ upper-triangulable matrices containing entries from $\mathbb{F}_q$ and possible unknowns is denoted by $\mathcal{U}^p_q$. 
\end{defn} 
 
We now explain the groupcast index coding problem setup.
 
\par An instance of the groupcast index coding problem consists of a sender with $m$ independent messages given by  $\mathcal{M} =\{{\bf{x}}_1,{\bf{x}}_2,\cdots,{\bf{x}}_{m}\}$, 
where ${\bf{x}}_i \in \mathbb{F}_q^{d \times 1}$, $i \in [m]$, and $d \geq 1$. There are $n$ receivers. The $j$th receiver knows $\mathcal{K}_j \subset \mathcal{M}$ (also known as its side-information) and wants $\mathcal{W}_j \subseteq  \mathcal{M} \setminus \mathcal{K}_j$, $j \in [n]$. 
Each message is demanded by at least one receiver. Without loss of generality, throughout the paper we assume that  $|\mathcal{W}_j|=1, \forall j \in [n]$. For a receiver demanding more than one message, we replace it by as many new receivers as the number of messages demanded by the original receiver, with each new receiver demanding a unique message which was demanded by the original receiver and having the same side-information as that of the original receiver. Hence, we assume that the $j$th receiver wants ${\bf{x}}_{f(j)}$, $j \in [n]$, where the mapping $f : [n] \rightarrow [m]$ gives the index of the wanted message. Let $\mathcal{K}=(\mathcal{K}_1, \mathcal{K}_2, \cdots, \mathcal{K}_n)$. Hence, we can describe an instance of the groupcast ICP using the quadruple $(m,n,\mathcal{K},f)$. The transmission is through a noiseless broadcast channel which carries symbols from $\mathbb{F}_q$.

An index code over $\mathbb{F}_{q}$ for an instance of the groupcast ICP, described by $(m,n,\mathcal{K},f)$, is an encoding function $\mathbb{E} : \mathbb{F}_{q}^{md \times 1} \rightarrow   \mathbb{F}_{q}^{r \times 1}$ such that there exists a decoding function $\mathbb{D}_{j}:\mathbb{F}_{q}^{(r+|\mathcal{K}_{j}|d) \times 1} \rightarrow   \mathbb{F}_{q}^{d \times 1}$ at $j$th receiver $\forall j \in [n]$, with ${\bf{x}}_{f(j)} = \mathbb{D}_{j}(\mathbb{E}({\bf{x}}),\mathcal{K}_j)$ for any realizations of $\mathcal{K}_j$ and ${\bf{x}} = ({\bf{x}}_1|...|...|{\bf{x}}_m)^T$. The sender transmits $\mathbb{E}(\bf{x})$ with codelength $r$. The smallest possible value of $r$ is called the optimal codelength of the problem. If the encoding function is linear, the index code is given by $\bf{G}\bf{x}$, where ${\bf{G}} \in \mathbb{F}^{r \times md}_{q}$ is called the encoding matrix for the given index code. With the encoding function being linear, if $d=1$, the code is said to be scalar linear, else it is said to be vector linear. In this paper, we only consider scalar linear codes. If $n=m$, the index coding problem (ICP) is called single unicast ICP (SUICP). For an SUICP, without loss of generality, we assume that the $j$th receiver wants ${\bf{x}}_j$, $j \in [n]$.  

Any groupcast ICP can be represented using a fitting matrix which was introduced in \cite{BY} and was defined again in \cite{PK} to include the groupcast problem. It contains unknown entries denoted by $x$. Each row of the fitting matrix represents a receiver and  each column represents a message.
  
\begin{defn}[Fitting Matrix, \cite{PK}]
	An $n \times m$ matrix ${\bf{F}}_x$ is called the fitting matrix of an ICP described by $(m,n,\mathcal{K},f)$, where the ($i,j$)th entry is given by 
	\[ 
	[{\bf{F}}_x]_{i,j}=
	\begin{cases}
	x               & if \ {\bf{x}}_j \in \mathcal{K}_i,\\
	1               & if \ j = f(i),\\
	0               & otherwise.
	\end{cases}
	\]
	$\forall$ $i \in [n]$, and $j \in [m]$.
\end{defn}

\par The minimum rank of ${\bf{F}}_x$ obtained by replacing the $x$'s in ${\bf{F}}_x$ with arbitrary values from $\mathbb{F}_q$ is called the minrank of ${\bf{F}}_x$ or that of the ICP described by $(m,n,\mathcal{K},f)$ over $\mathbb{F}_q$. It has been shown in \cite{DSC} that the optimal codelength of any scalar linear code over $\mathbb{F}_q$ is equal to the minrank of ${\bf{F}}_x$ over $\mathbb{F}_q$, denoted as $mrk_{q}({\bf{F}}_x)$. We say ${\bf{F}} \approx {\bf{F}}_x$ (${\bf{F}}$ completes ${\bf{F}}_x$ or equivalently ${\bf{F}}$ is a completion of ${\bf{F}}_x$), if ${\bf{F}}$ is obtained from ${\bf{F}}_x$ by replacing all the unknown elements by arbitrary elements from the given field of interest.

The notion of joint extensions of any finite number of SUICPs was introduced in \cite{CBSR}. We extend the definition to include joint extensions of any finite number of groupcast ICPs.

\begin{defn}[Joint Extension]
Consider $l$ ICPs where the $i$th ICP $\mathcal{I}_i$ is described using the fitting matrix ${\bf{F}}_x^{(i)}$, $i \in [l]$. An ICP $\mathcal{I}_{E}$ whose fitting matrix is given by ${\bf{F}}^{E}_{x}$ is called a jointly extended ICP (or simply a joint extension of $l$ ICPs), extended using ICPs $\mathcal{I}_1,...,\mathcal{I}_l$, if 
${\bf{F}}^{E}_{x}$ consists of all ${\bf{F}}_x^{(i)}$'s, $i \in [l]$, as its disjoint submatrices. The $l$ ICPs are called as the component problems (or sub-problems) of the jointly extended problem.
\end{defn}

In this paper, we study a special class of joint extensions of $m_B$ groupcast ICPs described as follows. Let the ICP $\mathcal{I}_B$ described by the $n_B \times m_B$ fitting matrix ${\bf{F}}_x^{B}$, be called the base problem. Let $l_j$ denote the number of occurrences of  $1$ in the $j$th column of ${\bf{F}}_x^{B}$, $j \in [m_B]$. The superscript and subscript $``B"$ stands for the base problem. Let the $i$th component ICP $\mathcal{I}_i$ be described by the $n_i \times m_i$ fitting matrix ${\bf{F}}_x^{(i)}$, $i \in [m_B]$. Then, we have the joint extension $\mathcal{I}_E$ of the $m_B$ component ICPs with respect to the base problem $\mathcal{I}_B$, described by the  $n_E \times m_E$ 
fitting matrix  ${\bf{F}}_x^{E}$ as given below in terms of its block matrices, where 
$n_E=\underset{j \in [m_B]}{\Sigma} n_jl_j$, and $m_E=\underset{i \in [m_B]}{\Sigma} m_i$.   
 \[ 
 \big[[{\bf{F}}^{E}_x]\big]_{i,j}=
 \begin{cases}
 {\bf{X}}          & if \ [{\bf{F}}^{B}_x]_{i,j} = x,\\
 {\bf{F}}_x^{(j)}  & if \ [{\bf{F}}^{B}_x]_{i,j} = 1,\\
 {\bf{0}}          & otherwise.
 \end{cases}
 \]
 $\forall$ $i \in [n_B]$, and $j \in [m_B]$. That is, ${\bf{F}}_x^{E}$ is obtained from  ${\bf{F}}_x^{B}$ by replacing the $1$'s in its $j$th column by ${\bf{F}}_x^{(j)}$, and replacing $x$'s and $0$'s by ${\bf{X}}$'s and ${\bf{0}}$'s of appropriate sizes respectively. The dependence of $\mathcal{I}_E$ on $(\mathcal{I}_i)_{i \in [m_B]}$ and $\mathcal{I}_B$ is denoted as $\mathcal{I}_E(\mathcal{I}_B;(\mathcal{I}_i)_{i \in [m_B]})$.  
 Throughout this paper, whenever we refer to blocks (or block matrices) of ${\bf{F}}_x^{E}$, we refer to the block matrices that are induced by ${\bf{F}}_x^{B}$ as seen in the construction of ${\bf{F}}_x^{E}$ from the fitting matrices of the component ICPs based on the fitting matrix of the base problem. The $i$th row of block matrices in ${\bf{F}}_x^{E}$ refers to the matrix $\big(\big[[{\bf{F}}^{E}_x]\big]_{i,1}| \big[[{\bf{F}}^{E}_x]\big]_{i,2}|\cdots|\cdots|\big[[{\bf{F}}^{E}_x]\big]_{i,m_B}\big)$, $i \in [n_B]$. For the sake of brevity, we refer to the $i$th row of block matrices of a matrix as its $i$th block-row.  Similarly, we refer to the $j$th column of block matrices, $j \in [m_B]$, and call it the $j$th block-column. 
  
\begin{rem}
In a recent work \cite{FYGL}, generalized lexicographic product of a finite number of side-information digraphs was introduced. The class of joint extensions introduced in this paper reduces to the generalized lexicographic product, if  the base ICP $\mathcal{I}_B$ and all the ICPs $(\mathcal{I}_i)_{i \in [m_B]}$ are SUICPs. When the base problem and all the component problems are SUICPs, the side-information digraph $G_0$ in the generalized lexicographic product in \cite{FYGL} corresponds to the base problem $\mathcal{I}_B$ stated in this paper.
\end{rem}

We illustrate the construction of the extended problem using two running examples, given the base problem and the component problems, in terms of the respective fitting matrices.  
\begin{exmp}
Consider $m_B=n_B=3$. The base problem $\mathcal{I}_{B}$ is described by the fitting matrix ${\bf{F}}_x^{B}$. Let  the component problems $(\mathcal{I}_{i})_{i \in [m_B]}$ be described by $({\bf{F}}_x^{(i)})_{i \in [m_B]}$ respectively.
\[
{\bf{F}}_x^{B}=
\left(
\begin{array}{ccc} 
1  & x & 0 \\
0  & 1 & x \\
x  & 0 & 1 \\
\end{array}
\right),
{\bf{F}}_x^{(1)}=
\left(
\begin{array}{cccc} 
1 & x & 0 & 0\\
0 & 1 & x & 0\\
0 & 0 & 1 & x\\
x & 0 & 0 & 1\\
\end{array}
\right),
\]	
\[
{\bf{F}}_x^{(2)}=
\left(
\begin{array}{cc} 
1 & x \\
x & 1 \\
\end{array}
\right),
{\bf{F}}_x^{(3)}=
\left(
\begin{array}{ccc} 
1 & 0 & x \\
x & 1 & 0 \\
x & x & 1 \\
\end{array}
\right).
\]
Observe that $n_1=m_1=4, n_2=m_2=2, n_3=m_3=3$, and $l_1=l_2=l_3=1$. All the problems involved in the construction of the extended problem are SUICPs.
The extended problem $\mathcal{I}_E(\mathcal{I}_B;(\mathcal{I}_i)_{i \in [m_B]})$ is described by ${\bf{F}}_x^{E}$ with $n_E=m_E=4+3+2=9$. The block matrices of ${\bf{F}}_x^{E}$ are indicated by the partition shown in ${\bf{F}}_x^{E}$.
\[
{\bf{F}}_x^{E}=
\left(
\begin{array}{cccc|cc|ccc} 
1 & x & 0 & 0 & x & x & 0 & 0 & 0 \\
0 & 1 & x & 0 & x & x & 0 & 0 & 0 \\
0 & 0 & 1 & x & x & x & 0 & 0 & 0 \\
x & 0 & 0 & 1 & x & x & 0 & 0 & 0 \\
\hline
0 & 0 & 0 & 0 & 1 & x & x & x & x \\
0 & 0 & 0 & 0 & x & 1 & x & x & x \\
\hline
x & x & x & x & 0 & 0 & 1 & 0 & x \\
x & x & x & x & 0 & 0 & x & 1 & 0 \\
x & x & x & x & 0 & 0 & x & x & 1 \\
\end{array}
\right).
\]  
\label{exmp1}
\end{exmp}

The following example illustrates the construction of an extended problem which is a groupcast problem, with the base problem also being a groupcast problem.

 \begin{exmp}
 	Consider $m_B=4,n_B=5$. The base problem and the component problems are described by the fitting matrices given below respectively.
 	\[
 	{\bf{F}}_x^{B}=
 	\left(
 	\begin{array}{cccc} 
 	1  & x & 0 & 0 \\
 	0  & x & 1 & 0 \\
    x  & 1 & 0 & 0 \\ 
 	0  & 0 & x & 1 \\
 	0  & 0 & 1 & x \\
 	\end{array}
 	\right),
 	{\bf{F}}_x^{(1)}=
 	\left(
 	\begin{array}{ccc} 
 	1 & x & 0\\
 	0 & x & 1\\
 	x & 1 & 0\\
 	1 & 0 & x\\
 	\end{array}
 	\right),
 	\]	
 	\[
 	{\bf{F}}_x^{(2)}=
 	\left(
 	\begin{array}{cc} 
 	1 & x \\
 	x & 1 \\
 	\end{array}
 	\right),
 	{\bf{F}}_x^{(3)}=
 	\left(
 	\begin{array}{ccc} 
 	1 & 0 & x \\
 	x & 1 & 0 \\
 	x & x & 1 \\
 	\end{array}
 	\right),
 	{\bf{F}}_x^{(4)}=
 	\left(
 	\begin{array}{c} 
 		1
 	\end{array}
 	\right)
 	.
 	\]
 	Observe that $n_1=4, m_1=3, n_2=m_2=2, n_3=m_3=3,$ and $ n_4=m_4=1$. Also, $l_1=l_2=l_4=1$ and $l_3=2$. Note that $\mathcal{I}_1$ is a groupcast problem.
 	The extended problem $\mathcal{I}_E(\mathcal{I}_B;(\mathcal{I}_i)_{i \in [m_B]})$ is described by ${\bf{F}}_x^{E}$ with $n_E=4+2+(2*3)+1=13,$ and $ m_E=3+2+3+1=9$.
 	\[
 	{\bf{F}}_x^{E}=
 	\left(
 	\begin{array}{ccc|cc|ccc|c} 
 	1 & x & 0 & x & x & 0 & 0 & 0 & 0 \\
 	0 & x & 1 & x & x & 0 & 0 & 0 & 0 \\
 	x & 1 & 0 & x & x & 0 & 0 & 0 & 0 \\
 	1 & 0 & x & x & x & 0 & 0 & 0 & 0 \\
 	\hline
 	0 & 0 & 0 & x & x & 1 & 0 & x & 0 \\
 	0 & 0 & 0 & x & x & x & 1 & 0 & 0 \\
 	0 & 0 & 0 & x & x & x & x & 1 & 0 \\
 	\hline
 	x & x & x & 1 & x & 0 & 0 & 0 & 0 \\
 	x & x & x & x & 1 & 0 & 0 & 0 & 0 \\
 	\hline
 	0 & 0 & 0 & 0 & 0 & x & x & x & 1 \\
 	\hline
 	0 & 0 & 0 & 0 & 0 & 1 & 0 & x & x \\
 	0 & 0 & 0 & 0 & 0 & x & 1 & 0 & x \\
 	0 & 0 & 0 & 0 & 0 & x & x & 1 & x \\
 	\end{array}
 	\right).
 	\]  
 	\label{exmp2}
 \end{exmp}
 
 \par The following notations are required for the construction of a larger index code from component index codes. Let $\mathcal{C}_1$ and $\mathcal{C}_2$ be  two codewords of length $l_1$ and $l_2$ respectively. $\mathcal{C}_1 + \mathcal{C}_2$ denotes the element-wise addition of $\mathcal{C}_1$ and $\mathcal{C}_2$ after zero-padding the shorter message at the least significant positions to match the length of the longer message. The resulting length of the codeword is  $max(l_1,l_2)$. For example, if $\mathcal{C}_1=1010$, and $\mathcal{C}_2=110$, then $\mathcal{C}_1 + \mathcal{C}_2 = 0110$.  $\mathcal{C}[a:b]$ denotes the vector obtained by picking the element from position $a$ to element with position $b$, starting from the most significant position of the codeword $\mathcal{C}$, with $a,b \in [l]$, $l$ being the length of $\mathcal{C}$. For example $\mathcal{C}_1[2:4]=010$.
 
The results presented in this paper hold for any finite field. But, we consider only $q=2$ (binary field) for simplicity. 
\section{Main Results}
\par In this section, we first provide a lower bound on the minrank of the jointly extended problem introduced in the previous section, in terms of those of the component problems and the upper-triangulable submatrices of the base problem. Then, we provide a code construction (not necessarily optimal) for a special class of the jointly extended problem, in terms of those of the component problems and the base problem. We then provide  necessary conditions for the optimality of the code construction.\\

The following lemma provides a lower bound on the minrank of $\mathcal{I}_E(\mathcal{I}_B;(\mathcal{I}_i)_{i \in [m_B]})$. The proof follows on similar lines as that of Lemma 4.2 in \cite{DSC1}.
We provide the proof for completeness. The set of all upper-triangulable submatrices of ${\bf{F}}^B_x$ is given by 
\[
\mathcal{U}_B=\{  {\bf{M}}_x: {\bf{M}}_x \prec {\bf{F}}^B_x, {\bf{M}}_x \in \mathcal{U}^{\mathcal{C}({\bf{M}}_x)}_q \}.
\]   

\begin{lem}[A lower bound]
	For a given jointly extended ICP $\mathcal{I}_E(\mathcal{I}_B;(\mathcal{I}_i)_{i \in [m_B]})$ we have
	\begin{gather*}
	mrk_q({\bf{F}}^{E}_x) \geq max \{\underset{s \in col({\bf{M}}_x,{\bf{F}}^{B}_x)} {\sum} mrk_q({\bf{F}}^{(s)}_x): {\bf{M}}_x \in \mathcal{U}_B\}.
	\end{gather*}
	\label{lowbnd}
\end{lem}
\begin{proof}
Consider the submatrix  ${\bf{M}}^E_x$ corresponding to an ${\bf{M}}_x$ constructed using the block matrices of ${\bf{F}}^E_x$ as follows. Let $(s_i,t_j)$ with $i \in [\mathcal{R}({\bf{M}}_x)]$ and $j \in [\mathcal{C}({\bf{M}}_x)]$, be an element of the cartesian product given by  $row({\bf{M}}_x,{\bf{F}}^B_x) \times col({\bf{M}}_x,{\bf{F}}^B_x)$ for any ${\bf{M}}_x \in \mathcal{U}_B$. Then, the $(i,j)$th block matrix of ${\bf{M}}^E_x$ is given by $\big[[{\bf{F}}^{E}_x]\big]_{s_i,t_j}$. From the construction of ${\bf{M}}^{E}_x$ and the fact that ${\bf{M}}_x$ is an upper-triangulable matrix, we see that   ${\bf{M}}^E_x$ can be written as a block upper-triangular matrix  ${\bf{U}}^{E}_x$, by permuting the rows and/or columns of block matrices of ${\bf{M}}^{E}_x$ using the same permutations that make ${\bf{M}}_x$ an upper-triangular matrix with all its diagonal entries being $1$. To prove the lemma, we find the minrank of ${\bf{M}}^{E}_x$ as ${\bf{M}}^{E}_x$ corresponds to a sub-problem of ${\bf{F}}^{E}_x$. Note that the minrank of any sub-problem is not greater than that of the original problem. Hence, we first provide an upper bound for  $mrk_q({\bf{M}}^{E}_x)$ and then provide a matching lower bound.
	
With all matrices ${\bf{F}}^{(t_j)}_x$, $t_j \in col({\bf{M}}_x,{\bf{F}}^B_x)$, $j \in [\mathcal{C}({\bf{M}}_x)]$, now being the diagonal block matrices of ${\bf{U}}^{E}_x$, if ${\bf{F}}^{(t_j)} \approx {\bf{F}}^{(t_j)}_x$, then the block  diagonal matrix ${\bf{D}}^{E}$ with its diagonal block matrices being  ${\bf{F}}^{(t_j)}$ in some order (due to the permutations applied on the rows and/or columns of block matrices of ${\bf{M}}^E_x$), we see that ${\bf{D}}^{E} \approx {\bf{U}}^{E}_x$. As $rk_q({\bf{D}}^{E})=\underset{t_j \in col({\bf{M}}_x,{\bf{F}}^{B}_x)} {\sum} rk_q({\bf{F}}^{(t_j)})$. Thus, we have
\begin{gather*}
mrk_q({\bf{M}}^{E}_x)= mrk_q({\bf{U}}^{E}_x)\leq rk_q({\bf{D}}^{E}) \\ = \underset{t_j \in col({\bf{M}}_x,{\bf{F}}^{B}_x)} {\sum} mrk_q({\bf{F}}^{(t_j)}_x),
\end{gather*}
where in the last equality, we take  ${\bf{F}}^{(t_j)} \approx {\bf{F}}^{(t_j)}_x$ such that $rk_q({\bf{F}}^{(t_j)})=mrk_q({\bf{F}}^{(t_j)}_x)$.	

Now, we provide a matching lower bound. If ${\bf{U}}^{E} \approx {\bf{U}}^{E}_x$, then ${\bf{U}}^{E}$ must be a block upper-triangular matrix. Note that the diagonal block entries $\big[[{\bf{U}}^{E}]\big]_{j',j'} \approx \big[[{\bf{U}}^{E}_x]\big]_{j',j'}$, and $\big[[{\bf{U}}^{E}_x]\big]_{j',j'}$  is equal to ${\bf{F}}^{(t_j)}_x$ for some $t_j \in col({\bf{M}}_x,{\bf{F}}^B_x)$ (due to the permutations applied on the rows and/or columns of block matrices of ${\bf{M}}^E_x$), $j,j' \in [\mathcal{C}({\bf{M}}_x)]$. Thus, we have 
\begin{gather*}
rk_q({\bf{U}}^{E}) \geq \underset{j' \in [\mathcal{C}({\bf{M}}_x)]} {\sum} rk_q(\big[[{\bf{U}}^{E}]\big]_{j',j'}) \\ \geq \underset{t_j \in col({\bf{M}}_x,{\bf{F}}^{B}_x)} {\sum} mrk_q({\bf{F}}^{(t_j)}_x),
\end{gather*}
which yields a matching lower bound by choosing ${\bf{U}}^{E}$ such that $rk_q({\bf{U}}^{E}) = mrk_q({\bf{U}}^{E}_x)$. This completes the proof.
\end{proof}
 
\begin{rem}
This lower bound resembles the MAIS (Maximum Acyclic Induced Subgraph) bound introduced in \cite{BY}, which is a lower bound on the minrank of the SUICP. However, the bound given in Lemma \ref{lowbnd} need not be equal to the MAIS bound for $\mathcal{I}_E(\mathcal{I}_B;(\mathcal{I}_i)_{i \in [m_B]})$. The submatrix of the fitting matrix of an SUICP corresponding to any maximum acyclic induced subgraph of the side-information digraph is upper-triangulable (as the subgraph is acyclic). Hence, we get the MAIS bound for the SUICP.  
\end{rem} 

\begin{rem}
The bound given in Lemma \ref{lowbnd} also has an operational significance in finding a lower bound on the minrank of the jointly extended problem using the minranks of some of the component sub-problems and the set $\mathcal{U}_B$ instead of directly computing lower bounds like the MAIS bound which is computation intensive.
\end{rem}
We illustrate the application of Lemma \ref{lowbnd} with two running examples. In the first example, all the problems involved are SUICPs. 
\begin{exmp}[Example \ref{exmp1} continued]
In Example \ref{exmp1}, we see that there are six upper-triangulable submatrices of ${\bf{F}}^B_x$, out of which considering all the $2 \times 2$ submatrices of ${\bf{F}}^B_x$ are sufficient to find the lower bound given in the lemma, as shown below. Note that $mrk_q({\bf{F}}^{(1)}_x)=3$, $mrk_q({\bf{F}}^{(2)}_x)=1$, and $mrk_q({\bf{F}}^{(3)}_x)=2$.
\begin{gather*}
{\bf{M}}_x^{(1)}=
\left(
\begin{array}{cc} 
1  & x \\
0  & 1 \\
\end{array}
\right), col({\bf{M}}_x^{(1)},{\bf{F}}^B_x) = \{1,2\},\\
 row({\bf{M}}_x^{(1)},{\bf{F}}^B_x) =\{1,2\},  \underset{s \in col({\bf{M}}_x^{(1)},{\bf{F}}^{B}_x)} {\sum} mrk_q({\bf{F}}^{(s)}_x)=4.
\end{gather*}
\begin{gather*}
{\bf{M}}_x^{(2)}=
\left(
\begin{array}{cc} 
1 & x \\
0 & 1 \\
\end{array}
\right),  col({\bf{M}}_x^{(2)},{\bf{F}}^B_x) = \{2,3\},\\
row({\bf{M}}_x^{(2)},{\bf{F}}^B_x) = \{2,3\}, \underset{s \in col({\bf{M}}_x^{(2)},{\bf{F}}^{B}_x)} {\sum} mrk_q({\bf{F}}^{(s)}_x)=3.
\end{gather*}
\begin{gather*}
{\bf{M}}_x^{(3)}=
\left(
\begin{array}{cc} 
1 & 0 \\
x & 1 \\
\end{array}
\right), col({\bf{M}}_x^{(3)},{\bf{F}}^B_x) = \{1,3\},\\
row({\bf{M}}_x^{(3)},{\bf{F}}^B_x) = \{1,3\}, \underset{s \in col({\bf{M}}_x^{(3)},{\bf{F}}^{B}_x)} {\sum} mrk_q({\bf{F}}^{(s)}_x)=5.
\end{gather*}
Hence, according to the lemma we have $mrk_q({\bf{F}}^{E}_x) \geq 5$. 
\label{exmp3}
\end{exmp}

In the following example, the base problem and a component problem are groupcast ICPs.
\begin{exmp}[Example \ref{exmp2} continued]
	In Example \ref{exmp2}, it can be easily seen that there are no $4 \times 4$ upper-triangulable submatrices of ${\bf{F}}^B_x$, since any combination of $4$ rows consists of either rows $1$ and $3$ or rows $4$ and $5$, which if present in a $4 \times 4$ submatrix, the submatrix is not  upper-triangulable. This is because the problem induced by rows ($1$ and $3$) and rows ($3$ and $4$) contain a cycle. Note that $mrk_q({\bf{F}}^{(1)}_x)=mrk_q({\bf{F}}^{(3)}_x)=2$ and $mrk_q({\bf{F}}^{(2)}_x)=mrk_q({\bf{F}}^{(4)}_x)=1$. Consider the submatrix given below which is upper-triangulable (There exist row and column permutations which make ${\bf{M}}_x$ upper-triangular). 
	\begin{gather*}
	{\bf{M}}_x=
	\left(
	\begin{array}{ccc} 
		1  & 0 & 0 \\
        0  & 1 & 0 \\ 
	    0  & x & 1 \\
	\end{array}
	\right), row({\bf{M}}_x,{\bf{F}}^B_x) = \{1,2,4\}, \\ col({\bf{M}}_x,{\bf{F}}^B_x) = \{1,3,4\},
	\underset{s \in col({\bf{M}}_x,{\bf{F}}^{B}_x)} {\sum} mrk_q({\bf{F}}^{(s)}_x)=5.
	\end{gather*}	
Hence, according to the lemma we have $mrk_q({\bf{F}}^{E}_x) \geq 5$. 
\label{exmp4}
\end{exmp}
 
The following lemma provides a code construction (not necessarily optimal) for a particular class of the jointly extended ICP $\mathcal{I}_E(\mathcal{I}_B;(\mathcal{I}_i)_{i \in [m_B]})$ using codes of the component problems $(\mathcal{I}_i)_{i \in [m_B]}$ and a code of the base problem $\mathcal{I}_B$.
\begin{lem}[An upper bound]
For a given jointly extended ICP $\mathcal{I}_E(\mathcal{I}_B;(\mathcal{I}_i)_{i \in [m_B]})$, let ${\bf{F}}^{(j)} \approx {\bf{F}}^{(j)}_x$, $\forall j \in [m_B]$, such that $r_j=rk_q({\bf{F}}^{(j)})$, where $r_j$ is not necessarily equal to $mrk_q({\bf{F}}^{(j)}_x)$. If there exists $(i)$ an upper-triangulable matrix ${\bf{M}}_x$ such that ${\bf{M}}_x \prec {\bf{F}}^B_x$ and $\{t_1, t_2, \cdots, t_{\mathcal{C}({\bf{M}}_x)}\}=col({\bf{M}}_x,{\bf{F}}^B_x)$, where $(t_i)_{i \in [m_B]}$ is a permutation of $[m_B]$ such that $r_{t_1} \geq r_{t_2} \geq \cdots \geq r_{t_{r_B}}$, and $r_{t_{r_B}} \geq r_{t_{i}}$ for $i \geq r_B=\mathcal{C}({\bf{M}}_x)$, and $(ii)$ there exists an ${\bf{F}}^{B} \approx {\bf{F}}^{B}_x$ with $r_B=rk_q({\bf{F}}^{B})$, where $r_B$ need not be necessarily equal to $mrk_q({\bf{F}}^{B}_x)$, such that the rows of ${\bf{F}}^{B}$ indexed by the numbers in $row({\bf{M}}_x,{\bf{F}}^B_x)$ are independent,  then there exists a scalar linear code of length  $\underset{j \in [r_B]} {\sum} r_{t_j}$.
\label{upbnd}
\end{lem}
\begin{proof}
We provide a construction of a scalar linear code with the stated codelength. For an easier visualization of the code construction and to alleviate the need of more notations, we permute the rows of the fitting matrices and the completions (given in the statement of the lemma) of all the component problems and the base problem as stated in the following. We also permute the columns of the fitting matrix and the completion of the base problem. Note that the permutation applied on the columns and/or rows of any given fitting matrix (mentioned above) is same as that applied on the columns and/or rows of the respective completion. Then, we provide a code construction for the jointly extended ICP $\mathcal{I'}_E(\mathcal{I'}_B;(\mathcal{I'}_i)_{i \in [m_B]})$, where the base problem ($\mathcal{I'}_B$) and all the component problems ($(\mathcal{I'}_i)_{i \in [m_B]}$) have fitting matrices obtained by the above mentioned permutations of rows and/or columns of the original fitting matrices. Note that there is no loss of generality in proving the lemma for the extended problem  $\mathcal{I'}_E$ as the permutations mentioned above rename the messages (in the case of column permutations) and receivers (in the case of row permutations), which do not change the extended problem, the base problem and the component problems.

The rows of ${\bf{F}}^{(j)}$ and ${\bf{F}}^{(j)}_x$ are permuted with the same permutation such that the first $r_j$ rows of ${\bf{F}}^{(j)}$ are independent and span $\langle {\bf{F}}^{(j)} \rangle$ (the row space of ${\bf{F}}^{(j)}$), $\forall j \in [m_B]$. Note that such a permutation exists as $r_j=rk_q({\bf{F}}^{(j)})$, $\forall j \in [m_B]$. Let $\{s_1,s_2,\cdots,s_{r_B}\}=row({\bf{M}}_x,{\bf{F}}^{B}_x)$, where $(s_i)_{i \in [n_B]}$ is a permutation of $[n_B]$ such that $r_{s_1} \geq r_{s_2} \geq \cdots \geq r_{s_{r_B}}$, and $r_{s_{r_B}} \geq r_{s_i}$ for $i \geq r_B$,   where $r_{s_i}=rk_q({\bf{F}}^{(j)})$ such that $({\bf{F}}^{B}_x)_{s_i,j}=1$, $j \in [m_B]$.
The rows of ${\bf{F}}^{B}$ and ${\bf{F}}^{B}_x$ are permuted with the same permutation such that the rows of ${\bf{F}}^{B}$ indexed by the elements in $row({\bf{M}}_x,{\bf{F}}^B_x)$ are mapped to the first $r_B$ rows of ${\bf{F}}^{B}$ such that the row indexed by ${s_i}$ is mapped to the row indexed by $i$, $i \in [r_B]$. The columns of ${\bf{F}}^{B}$ and ${\bf{F}}^{B}_x$ are also permuted with the same permutation such that the columns of ${\bf{F}}^{B}$ indexed by the elements in $col({\bf{M}}_x,{\bf{F}}^B_x)$ are mapped to the first $r_B$ columns of ${\bf{F}}^{B}$. 
Now, consider the  fitting matrix ${\bf{F}}^{E}_x$ of the jointly extended ICP $\mathcal{I}_E(\mathcal{I}_B;(\mathcal{I}_i)_{i \in [m_B]})$, with the fitting  matrices of the base problem (${\bf{F}}^{B}_x$) and the component problems ($({\bf{F}}^{(j)}_x)_{j \in [m_B]}$) obtained after the above mentioned permutations. (Observe that we have not renamed the problems, fitting matrices, and their completions obtained after the permutation to have brevity in the notation. Due to the permutation, we now have $(t_i)_{i \in r_B}$ (defined in the statement of the theorem) mapped to $[r_B]$ in some order, which form the new $(t_i)_{i \in r_B}$. Hence, we now have $r_{r_B} \geq r_i, \forall i \geq r_B, i \in [m_B]$.) 
We now provide a completion of ${\bf{F}}^{E}_x$ and show that it is a valid completion. Then, we prove that the codelength obtained by such a completion is  $\underset{j \in [r_B]} {\sum} r_{j}$. 

Let ${\bf{P}}^{(j)}$ be an $(n_j-r_j) \times r_j$ matrix such that the last $(n_j-r_j)$ rows of  ${\bf{F}}^{(j)}$ are given by  ${\bf{P}}^{(j)}{\bf{F}}^{(j)}_{[[r_j]]}$, $\forall j \in [m_B]$. Let ${\bf{P}}^{B}$ be an $(n_B-r_B) \times r_B$ matrix such that the last $(n_B-r_B)$ rows of  ${\bf{F}}^{B}$ are given by  ${\bf{P}}^{B}{\bf{F}}^{B}_{[[r_B]]}$. Complete the first $r_i$ consecutive rows of the $i$th row of block matrices of ${\bf{F}}^{E}_x$, $i \in [r_B]$, with $({\bf{{F}}}^{B}_{i,1}{\bf{\hat{F}}}^{(i,1)}|{\bf{{F}}}^{B}_{i,2}{\bf{\hat{F}}}^{(i,2)}|\cdots|\cdots|{\bf{{F}}}^{B}_{i,{m_B}}{\bf{\hat{F}}}^{(i,{m_B})})$, where ${\bf{\hat{F}}}^{(i,j)}$ is given as 
    \begin{equation}
    {\bf{\hat{F}}}^{(i,j)} =
     \begin{cases}
     {\bf{F}}^{(j)}_{[[r_i]]} & if \ r_j \geq r_i, \\
     \left( \begin{array}{c} {\bf{F}}^{(j)}_{[[r_j]]}\\
     \hline
     {\bf{0}}_{(r_i-r_j) \times m_j}\\
     \end{array}
     \right) & if \ r_j < r_i.
     \end{cases}
     \label{Fcap}
     \end{equation}
     $j \in [m_B]$. Note that ${\bf{{F}}}^{B}_{i,j}$ is a scalar for $i \in [n_B]$, and $j \in [m_B]$.
     Complete the next $n_i-r_i$ consecutive rows of the $i$th row of block matrices of ${\bf{F}}^{E}_x$ with ${\bf{P}}^{(i)}({\bf{{F}}}^{B}_{i,1}{\bf{\hat{F}}}^{(i,1)}|{\bf{{F}}}^{B}_{i,2}{\bf{\hat{F}}}^{(i,2)}|\cdots|\cdots|{\bf{{F}}}^{B}_{i,{m_B}}{\bf{\hat{F}}}^{(i,{m_B})})$. Note that these $n_i-r_i$ consecutive rows are in the row space of the first $r_i$ consecutive rows of the $i$th row of block matrices. Note also that this is a valid completion of the first $r_B$ rows of block matrices of ${\bf{{F}}}^{E}_{x}$. Consider the matrix ${\bf{\hat{F}}}^{E}$ obtained by stacking the first $r_i$ rows of the $i$th row of block matrices of the completion one above the other starting from $i=1$, for $i \in [r_B]$. 
     
     From the fact that ${\bf{M}}_x$ is upper triangulable and hence by some permutations of the rows and columns of block matrices in ${\bf{\hat{F}}}^{E}$ it can be made upper-triangular. The resulting matrix has ${\bf{\hat{F}}}^{(i,i)}$, $i \in [r_B]$, as its block diagonal matrices  which are full rank matrices. By appropriate row reductions of the rows of block matrices, it can be easily seen that the rank of this matrix is $\underset{j \in [r_B]} {\sum} r_{j}$.  Now, we complete the remaining rows of block matrices of ${\bf{F}}^{E}_x$ and provide a completion which is in the row space of the matrix ${\bf{\hat{F}}}^{E}$ described above.
     
     Consider any $i$th row of block matrices for $i \in [n_B] \setminus [r_B]$. Let $j_i \in [m_B]$ be such that  $[{\bf{F}}^{B}_x]_{i,j_i}=1$. Complete the  first $r_{j_i}$ rows of any $(i,j)$th block matrix $\big[[{\bf{F}}^{E}_x]\big]_{i,j}$, $j \in [m_B]$, using  $\big[[{\bf{F}}^{E}]\big]_{i,j}=\underset{k \in [r_B]} {\sum} {\bf{P}}^{B}_{i,k}{\bf{F}}^{B}_{k,j}{\bf{\hat{F}}}^{(k,j)}_{[[r_{j_i}]]} = (\underset{k \in [r_B]} {\sum} {\bf{P}}^{B}_{i,k}{\bf{F}}^{B}_{k,j}){\bf{\hat{F}}}^{(1,j)}_{[[r_{j_i}]]}$, as ${\bf{\hat{F}}}^{(k,j)}_{[[r_{j_i}]]}={\bf{\hat{F}}}^{(1,j)}_{[[r_{j_i}]]}$ for any $k \in [r_B]$ from the definition of ${\bf{\hat{F}}}^{(i,j)}$. It can be easily verified that this is a valid completion of any $(i,j)$th block matrix $\big[[{\bf{F}}^{E}_x]\big]_{i,j}$.  If $j=j_i$, then $\underset{k \in [r_B]} {\sum} {\bf{P}}^{B}_{i,k}{\bf{F}}^{B}_{k,j}=1$, as ${\bf{F}}^{B} \approx {\bf{F}}^{B}_x$, and hence $\big[[{\bf{F}}^{E}]\big]_{i,j} \approx \big[[{\bf{F}}^{E}_x]\big]_{i,j}$. If $j \neq j_i$ such that $[{\bf{F}}^{B}_x]_{i,j}=0$, we know that $\underset{k \in [r_B]} {\sum} {\bf{P}}^{B}_{i,k}{\bf{F}}^{B}_{k,j}=0$ (as ${\bf{F}}^{B} \approx {\bf{F}}^{B}_x$ as before), and hence again $\big[[{\bf{F}}^{E}]\big]_{i,j} \approx \big[[{\bf{F}}^{E}_x]\big]_{i,j}$. Similarly it can be verified that the first $r_{j_i}$ rows of the $i$th row of block matrices is in the row space of ${\bf{\hat{F}}}^{E}$. The remaining $n_i-r_{j_i}$ rows of any $i$th row of block matrices for $i \in [n_B] \setminus [r_B]$ are filled by pre-multiplying the first $r_{j_i}$ rows of the $i$th row of block matrices by ${\bf{P}}^{(j_i)}$. It can be easily verified that this is a valid completion and the completion is in the row space of the first $r_{j_i}$ rows, which is in turn in the row space of ${\bf{\hat{F}}}^{E}$. This completes the proof.    
\end{proof}

We illustrate the use of Lemma \ref{upbnd} using a running example. 

\begin{exmp}(Example \ref{exmp1} continued)
	Consider the completions ${\bf{F}}^{(i)}$ of ${\bf{F}}^{(i)}_x$, $i \in [3]$ as given below with $r_1=rk_q({\bf{F}}^{(1)})=3$, $r_2=rk_q({\bf{F}}^{(2)})=1$, and $r_3=rk_q({\bf{F}}^{(3)})=3$. 
		\[
		{\bf{F}}^{(1)}=
		\left(
		\begin{array}{cccc} 
		1 & 1 & 0 & 0\\
		0 & 1 & 1 & 0\\
		0 & 0 & 1 & 1\\
		1 & 0 & 0 & 1\\
		\end{array}
		\right),~
		{\bf{F}}^{(2)}=
		\left(
		\begin{array}{cc} 
		1 & 1 \\
		1 & 1 \\
		\end{array}
		\right),
		\]
		\[
		{\bf{F}}^{(3)}=
		\left(
		\begin{array}{ccc} 
		1 & 0 & 0 \\
		0 & 1 & 0 \\
		0 & 0 & 1 \\
		\end{array}
		\right),~
		{\bf{P}}^{(1)}=
		\left(
		\begin{array}{ccc} 
		1 & 1 & 1 \\
		\end{array}
		\right),~
		{\bf{P}}^{(2)}=
		\left(
		\begin{array}{c} 
		1 \\
		\end{array}
		\right).
		\]
	Note that only the completions ${\bf{F}}^{(1)}$ and ${\bf{F}}^{(2)}$ correspond to optimal codes as $r_1$ and $r_2$ are equal to $mrk_q({\bf{F}}^{(1)}_x)$ and $mrk_q({\bf{F}}^{(2)}_x)$ respectively. Also, $r_1 \geq r_3 \geq r_2$. Hence, letting $t_1=1$, $t_2=3$, and $t_3=2$, and taking the third submatrix ${\bf{M}}_x^{(3)}$ of ${\bf{F}}^B_x$ given in Example \ref{exmp3} (given below for easy reference), we see that the condition $(i)$ given in Lemma \ref{upbnd} is satisfied.
	\begin{gather*}
	{\bf{M}}_x^{(3)}=
	\left(
	\begin{array}{cc} 
	1 & 0 \\
	x & 1 \\
	\end{array}
	\right), row({\bf{M}}_x^{(3)},{\bf{F}}^B_x) = \{1,3\},\\ col({\bf{M}}_x^{(3)},{\bf{F}}^B_x) = \{1,3\}.
	\end{gather*}
Note that by taking ${\bf{F}}^{B} \approx {\bf{F}}^{B}_x$	as given below, condition $(ii)$ given in Lemma \ref{upbnd} is also satisfied.
\[
{\bf{F}}^{B}=
\left(
\begin{array}{ccc} 
1  & 1 & 0 \\
0  & 1 & 1 \\
1  & 0 & 1 \\
\end{array}
\right), r_B=rk_q({\bf{F}}^{B})=mrk_q({\bf{F}}^{B}_x)=2.
\]
Now we complete the fitting matrix ${\bf{F}}^E_x$ as given in Lemma \ref{upbnd} as shown below. Note that double lines (in ${\bf{F}}^E$) used for partitioning correspond to the block matrices of ${\bf{F}}^E_x$. The single lines correspond to the construction given in Lemma \ref{upbnd}. The matrix ${\bf{\hat{F}}}^{E}$ is also shown below.
\[
{\bf{F}}^{E}=
\left(
\begin{array}{cccc||cc||ccc} 
1 & 1 & 0 & 0 & 1 & 1 & 0 & 0 & 0 \\
0 & 1 & 1 & 0 & 0 & 0 & 0 & 0 & 0 \\
0 & 0 & 1 & 1 & 0 & 0 & 0 & 0 & 0 \\
\hline
1 & 0 & 0 & 1 & 1 & 1 & 0 & 0 & 0 \\
\hline
\hline
0 & 0 & 0 & 0 & 1 & 1 & 1 & 0 & 0 \\
\hline
0 & 0 & 0 & 0 & 1 & 1 & 1 & 0 & 0 \\
\hline
\hline
1 & 1 & 0 & 0 & 0 & 0 & 1 & 0 & 0 \\
0 & 1 & 1 & 0 & 0 & 0 & 0 & 1 & 0 \\
0 & 0 & 1 & 1 & 0 & 0 & 0 & 0 & 1 \\
\end{array}
\right).
\]  
\[
{\bf{\hat{F}}}^{E}=
\left(
\begin{array}{cccc||cc||ccc} 
1 & 1 & 0 & 0 & 1 & 1 & 0 & 0 & 0 \\
0 & 1 & 1 & 0 & 0 & 0 & 0 & 0 & 0 \\
0 & 0 & 1 & 1 & 0 & 0 & 0 & 0 & 0 \\
\hline
\hline
1 & 1 & 0 & 0 & 0 & 0 & 1 & 0 & 0 \\
0 & 1 & 1 & 0 & 0 & 0 & 0 & 1 & 0 \\
0 & 0 & 1 & 1 & 0 & 0 & 0 & 0 & 1 \\
\end{array}
\right).
\]  
Consider the message set given by  $\mathcal{M}=\{{\bf{x}}_1={\bf{x}}_1^{(1)},{\bf{x}}_2={\bf{x}}_1^{(2)},{\bf{x}}_3={\bf{x}}_1^{(3)},{\bf{x}}_4={\bf{x}}_1^{(4)},{\bf{x}}_5={\bf{x}}_2^{(1)},{\bf{x}}_6={\bf{x}}_2^{(2)},{\bf{x}}_7={\bf{x}}_3^{(1)},{\bf{x}}_8={\bf{x}}_3^{(2)},{\bf{x}}_9={\bf{x}}_3^{(3)}\}$. The index code is given by ${\bf{\hat{F}}}^{E}{\bf{x}}$, where $x=({\bf{x}}_1~{\bf{x}}_2~{\bf{x}}_3~{\bf{x}}_4~{\bf{x}}_5~{\bf{x}}_6~{\bf{x}}_7~{\bf{x}}_8~{\bf{x}}_9)^T$. Hence, the code is given by  $\mathcal{C}_E=({\bf{x}}_1+{\bf{x}}_2+{\bf{x}}_5+{\bf{x}}_6$, $~{\bf{x}}_2+{\bf{x}}_3$, $~{\bf{x}}_3+{\bf{x}}_4$, $~{\bf{x}}_1+{\bf{x}}_2+{\bf{x}}_7$, $~{\bf{x}}_2+{\bf{x}}_3+{\bf{x}}_8$, $~{\bf{x}}_3+{\bf{x}}_4+{\bf{x}}_9)$. It can be easily verified that all receivers can decode their wanted messages using their side-information and $\mathcal{C}_E$.

Note that the code given by ${\bf{F}}^{B}$ for the base problem is $\mathcal{C}_B=({\bf{x}}^{(B)}_1+{\bf{x}}^{(B)}_2,~ {\bf{x}}^{(B)}_1+{\bf{x}}^{(B)}_3)$, where the message set for the base problem is given by $\mathcal{M}_B=\{{\bf{x}}^{(B)}_1,{\bf{x}}^{(B)}_2,{\bf{x}}^{(B)}_3\}$. Considering the codes of the component problems given by  $\mathcal{C}_1=({\bf{x}}_1^{(1)}+{\bf{x}}_2^{(1)},~ {\bf{x}}_2^{(1)}+{\bf{x}}_3^{(1)},~{\bf{x}}_3^{(1)}+{\bf{x}}_4^{(1)})$,   $\mathcal{C}_2=({\bf{x}}_1^{(2)}+{\bf{x}}_2^{(2)})$, and $\mathcal{C}_3=({\bf{x}}_1^{(3)},~ {\bf{x}}_2^{(3)},~{\bf{x}}_3^{(3)})$, we see that the code $\mathcal{C}_E$ can also be written as $(\mathcal{C}_1 + \mathcal{C}_2,\mathcal{C}_1 + \mathcal{C}_3 )$. This shows the dependence of the code $\mathcal{C}_E$ on those of the base problem and the component problems. In the code $\mathcal{C}_B$, ${\bf{x}}^{(B)}_i$ is replaced by  $\mathcal{C}_i$, for $i \in [m_B]$, to obtain $\mathcal{C}_E$.
\end{exmp}

We now provide an example with a given groupcast ICP.
\begin{exmp}
Consider the groupcast ICP given by the fitting matrix shown below with $n_E=14$ and $m_E=11$. 
\[{\bf{F}}^{E}_x=
\left(
\begin{array}{ccccccccccc} 
1 & 0 & 0 & 0 & x & x & x & 0 & 0 & 0 & 0\\
0 & 1 & x & 0 & x & x & x & 0 & 0 & 0 & 0\\
0 & 0 & 1 & 0 & x & x & x & 0 & 0 & 0 & 0\\
x & 0 & 0 & 1 & x & x & x & 0 & 0 & 0 & 0\\
0 & 0 & 0 & 0 & 1 & 0 & 0 & 0 & 0 & x & x\\
0 & 0 & 0 & 0 & 0 & 1 & 0 & 0 & 0 & x & x\\
0 & 0 & 0 & 0 & x & 0 & 1 & 0 & 0 & x & x\\
0 & 0 & 0 & 0 & 0 & 0 & 0 & 1 & 0 & x & x\\	
0 & 0 & 0 & 0 & 0 & 0 & 0 & x & 1 & x & x\\	
x & x & x & x & 1 & 0 & 0 & 0 & 0 & 0 & 0\\
x & x & x & x & 0 & 1 & 0 & 0 & 0 & 0 & 0\\
x & x & x & x & x & 0 & 1 & 0 & 0 & 0 & 0\\
0 & 0 & 0 & 0 & x & x & x & 0 & 0 & 1 & x\\
0 & 0 & 0 & 0 & x & x & x & 0 & 0 & 0 & 1\\
\end{array}
\right).
\]
It can be easily identified that this problem is a jointly extended ICP introduced in this paper. The fitting matrices of the component problems and the base problem are given below.
\[{\bf{F}}^{B}_x=
\left(
\begin{array}{cccc} 
1 & x & 0 & 0\\
0 & 1 & 0 & x\\
0 & 0 & 1 & x \\
x & 1 & 0 & 0 \\
0 & x & 0 & 1\\
\end{array}
\right), ~{\bf{F}}^{(1)}_x=
\left(
\begin{array}{cccc} 
1 & 0 & 0 & 0\\
0 & 1 & x & 0\\
0 & 0 & 1 & 0\\
x & 0 & 0 & 1\\
\end{array}
\right),
\]
 \[{\bf{F}}^{(2)}_x=
 \left(
 \begin{array}{ccc} 
 1 & 0 & 0 \\
 0 & 1 & 0 \\
 x & 0 & 1 \\
 \end{array}
 \right), ~{\bf{F}}^{(3)}_x=
 \left(
 \begin{array}{cc} 
 1 & 0\\
 x & 1\\
 \end{array}
 \right),
 \]
 \[ {\bf{F}}^{(4)}_x=
 \left(
 \begin{array}{cc} 
 1 & x\\
 0 & 1\\
 \end{array}
 \right).
 \]
Note that $mrk_q({\bf{F}}^{(1)}_x)=4$, $mrk_q({\bf{F}}^{(2)}_x)=3$,
$mrk_q({\bf{F}}^{(3)}_x)=mrk_q({\bf{F}}^{(4)}_x)=2$, and $mrk_q({\bf{F}}^{B}_x)=3$. As in Example \ref{exmp4}, we find that there are no $4 \times 4$ upper-triangulable submatrices of ${\bf{F}}^{E}_x$. By enumerating all possible $3 \times 3$ submatrices, we see that there are three $3 \times 3$ upper-triangulable submatrices as given below.
\begin{gather*}
{\bf{M}}_x^{(1)}=
\left(
\begin{array}{ccc} 
1 & x & 0 \\
0 & 1 & 0 \\ 
0 & 0 & 1 \\
\end{array}
\right),~ row({\bf{M}}_x^{(1)},{\bf{F}}^B_x) = \{1,2,3\}, \\ col({\bf{M}}_x^{(1)},{\bf{F}}^B_x) = \{1,2,3\}, \underset{s \in col({\bf{M}}_x^{(1)},{\bf{F}}^{B}_x)} {\sum} mrk_q({\bf{F}}^{(s)}_x)=9.
\end{gather*}
\begin{gather*}
{\bf{M}}_x^{(2)}=
\left(
\begin{array}{ccc} 
1 & 0 & 0 \\
0 & 1 & x \\ 
0 & 0 & 1 \\
\end{array}
\right),~ row({\bf{M}}_x^{(2)},{\bf{F}}^B_x) = \{1,3,5\}, \\ col({\bf{M}}_x^{(2)},{\bf{F}}^B_x) = \{1,3,4\}, \underset{s \in col({\bf{M}}_x^{(2)},{\bf{F}}^{B}_x)} {\sum} mrk_q({\bf{F}}^{(s)}_x)=8.
\end{gather*}
\begin{gather*}
{\bf{M}}_x^{(3)}=
\left(
\begin{array}{ccc} 
0 & 1 & x \\
1 & 0 & 0 \\ 
x & 0 & 1 \\
\end{array}
\right),~ row({\bf{M}}_x^{(3)},{\bf{F}}^B_x) = \{3,4,5\}, \\ col({\bf{M}}_x^{(3)},{\bf{F}}^B_x) = \{2,3,4\}, \underset{s \in col({\bf{M}}_x^{(3)},{\bf{F}}^{B}_x)} {\sum} mrk_q({\bf{F}}^{(s)}_x)=7.
\end{gather*}
Consider the completion of ${\bf{F}}^B_x$ given below. Observe that the first three rows are independent and span $\langle {\bf{F}}^B \rangle$. Note that this choice of ${\bf{F}}^{B}$ and ${\bf{M}}_x^{(1)}$ satisfy conditions $(i)$ and $(ii)$ given in the lemma.
\[{\bf{F}}^{B}=
\left(
\begin{array}{cccc} 
1 & 1 & 0 & 0\\
0 & 1 & 0 & 1\\
0 & 0 & 1 & 1 \\
1 & 1 & 0 & 0 \\
0 & 1 & 0 & 1\\
\end{array}
\right).
\]
We complete ${\bf{F}}^{E}_x$ as given in the lemma. The encoding matrix ${\bf{\hat{F}}}^{E}$ obtained by this completion  is also given below. Observe that the codelength is $4+3+2=9$ as stated by the lemma.
\[{\bf{F}}^{E}=
\left(
\begin{array}{cccc||ccc||cc||cc} 
1 & 0 & 0 & 0 & 1 & 0 & 0 & 0 & 0 & 0 & 0\\
0 & 1 & 0 & 0 & 0 & 1 & 0 & 0 & 0 & 0 & 0\\
0 & 0 & 1 & 0 & 0 & 0 & 1 & 0 & 0 & 0 & 0\\
0 & 0 & 0 & 1 & 0 & 0 & 0 & 0 & 0 & 0 & 0\\
\hline
\hline
0 & 0 & 0 & 0 & 1 & 0 & 0 & 0 & 0 & 1 & 0\\
0 & 0 & 0 & 0 & 0 & 1 & 0 & 0 & 0 & 0 & 1\\
0 & 0 & 0 & 0 & 0 & 0 & 1 & 0 & 0 & 0 & 0\\
\hline
\hline
0 & 0 & 0 & 0 & 0 & 0 & 0 & 1 & 0 & 0 & 0\\	
0 & 0 & 0 & 0 & 0 & 0 & 0 & 0 & 1 & 0 & 0\\	
\hline
\hline
1 & 0 & 0 & 0 & 1 & 0 & 0 & 0 & 0 & 0 & 0\\
0 & 1 & 0 & 0 & 0 & 1 & 0 & 0 & 0 & 0 & 0\\
0 & 0 & 1 & 0 & 0 & 0 & 1 & 0 & 0 & 0 & 0\\
\hline
\hline
0 & 0 & 0 & 0 & 1 & 0 & 0 & 0 & 0 & 1 & 0\\
0 & 0 & 0 & 0 & 0 & 1 & 0 & 0 & 0 & 0 & 1\\
\end{array}
\right).
\]
\[
{\bf{\hat{F}}}^{E}=
\left(
\begin{array}{cccc||ccc||cc||cc} 
1 & 0 & 0 & 0 & 1 & 0 & 0 & 0 & 0 & 0 & 0\\
0 & 1 & 0 & 0 & 0 & 1 & 0 & 0 & 0 & 0 & 0\\
0 & 0 & 1 & 0 & 0 & 0 & 1 & 0 & 0 & 0 & 0\\
0 & 0 & 0 & 1 & 0 & 0 & 0 & 0 & 0 & 0 & 0\\
\hline
\hline
0 & 0 & 0 & 0 & 1 & 0 & 0 & 0 & 0 & 1 & 0\\
0 & 0 & 0 & 0 & 0 & 1 & 0 & 0 & 0 & 0 & 1\\
0 & 0 & 0 & 0 & 0 & 0 & 1 & 0 & 0 & 0 & 0\\
\hline
\hline
0 & 0 & 0 & 0 & 0 & 0 & 0 & 1 & 0 & 0 & 0\\	
0 & 0 & 0 & 0 & 0 & 0 & 0 & 0 & 1 & 0 & 0\\
\end{array}
\right).
\] 
\label{exmp6}
\end{exmp}
Now, we state and prove the main result of this section, which establishes the minrank of a special class of jointly extended problems identified in this paper.
 
\begin{thm}
		For a given jointly extended ICP $\mathcal{I}_E(\mathcal{I}_B;(\mathcal{I}_i)_{i \in [m_B]})$, with $r_j=mrk_q({\bf{F}}^{(j)}_x)$, $\forall j \in [m_B]$. Let $r_{t_1} \geq r_{t_2} \geq \cdots \geq r_{t_{m_B}}$, where $t_j,j \in [m_B]$. If there exists $(i)$ an upper-triangulable matrix ${\bf{\hat{M}}}_x \prec {\bf{F}}^B_x$ such that 		
		\begin{gather*}
		\underset{s \in col({\bf{\hat{M}}}_x,{\bf{F}}^{B}_x)} {\Sigma} mrk_q({\bf{F}}^{(s)}_x) = \\ max \{\underset{s \in col({\bf{M}}_x,{\bf{F}}^{B}_x)} {\Sigma} mrk_q({\bf{F}}^{(s)}_x): {\bf{M}}_x \in \mathcal{U}_B\},
		\end{gather*}
		where $col({\bf{\hat{M}}}_x,{\bf{F}}^B_x) = \{t_1, t_2, \cdots, t_{\mathcal{C}({\bf{\hat{M}}}_x)}\}$, and there exists $(ii)$ an ${\bf{F}}^{B} \approx {\bf{F}}^{B}_x$ with   $r_B=rk_q({\bf{F}}^{B})=\mathcal{C}({\bf{\hat{M}}}_x)$, such that the rows of ${\bf{F}}^{B}$ indexed by the numbers in $row({\bf{\hat{M}}}_x,{\bf{F}}^B_x)$ are independent, then we have  $mrk_q({\bf{F}}^{E}_x)=\underset{j \in [r_B]} {\sum} r_{t_j}$.
		\label{thm1}
\end{thm}
\begin{proof}
The proof follows directly from Lemmas \ref{lowbnd} and \ref{upbnd}, which provide a lower bound and the matching upper bound respectively, with the conditions stated in the theorem.
\end{proof}

The optimality of the scalar linear code given in Example \ref{exmp6} follows from this theorem. We provide another example to illustrate the use of the theorem.
\begin{exmp}
Consider the groupcast ICP given by the fitting matrix shown below with $n_E=15$ and $m_E=14$. The fitting matrices of the component problems and the base problem are also identified given below.
\[{\bf{F}}^{E}_x=
\left(
\begin{array}{ccccccccccccc} 
1 & x & x & x & x & x & x & x & x & 0 & 0 & 0 & 0\\
0 & 1 & x & x & x & x & x & x & x & 0 & 0 & 0 & 0\\
0 & 0 & 1 & x & 0 & x & x & x & x & x & x & x & x\\
0 & 0 & 0 & 1 & x & x & x & x & x & x & x & x & x\\
0 & 0 & x & 0 & 1 & x & x & x & x & x & x & x & x\\
x & x & 0 & 0 & 0 & 1 & x & x & 0 & x & x & x & x\\
x & x & 0 & 0 & 0 & 0 & 1 & x & x & x & x & x & x\\
x & x & 0 & 0 & 0 & x & 0 & 1 & x & x & x & x & x\\	
x & x & 0 & 0 & 0 & x & x & 0 & 1 & x & x & x & x\\	
x & x & x & x & x & 0 & 0 & 0 & 0 & 1 & x & 0 & 0\\
x & x & x & x & x & 0 & 0 & 0 & 0 & 0 & 1 & 0 & x\\
x & x & x & x & x & 0 & 0 & 0 & 0 & 0 & 0 & 1 & x\\
x & x & x & x & x & 0 & 0 & 0 & 0 & x & 1 & 0 & 0\\
x & x & x & x & x & 0 & 0 & 0 & 0 & 0 & x & 0 & 1\\
\end{array}
\right).
\]
\[{\bf{F}}^{B}_x={\bf{F}}^{(3)}_x=
\left(
\begin{array}{cccc}
1 & x & x & 0\\
0 & 1 & x & x\\
x & 0 & 1 & x\\
x & x & 0 & 1\\
\end{array}
\right),
{\bf{F}}^{(1)}_x=
\left(
\begin{array}{cc}
1 & x\\ 
0 & 1\\ 
\end{array}
\right),
\]	
\[{\bf{F}}^{(2)}_x=
\left(
\begin{array}{ccc}
1 & x & 0\\
0 & 1 & x\\
x & 0 & 1\\
\end{array}
\right),
{\bf{F}}^{(4)}_x=
\left(
\begin{array}{cccc}
1 & x & 0 & 0\\ 
0 & 1 & 0 & x\\
0 & 0 & 1 & x\\
x & 1 & 0 & 0\\
0 & x & 0 & 1\\ 
\end{array}
\right).
\]	
Note that $mrk_q({\bf{F}}^{(1)}_x)=mrk_q({\bf{F}}^{(2)}_x)=mrk_q({\bf{F}}^{(3)}_x)=2$, and
$mrk_q({\bf{F}}^{(4)}_x)=3$. Observe that there are no upper-triangulable matrices of size $3 \times 3$ in ${\bf{F}}^{B}_x$. Consider the following upper-triangulable submatrix. 
\begin{gather*}
{\bf{M}}_x=
\left(
\begin{array}{cc} 
1 & x\\
0 & 1\\ 
\end{array}
\right),~ row({\bf{M}}_x,{\bf{F}}^B_x) = \{3,4\}, \\ col({\bf{M}}_x,{\bf{F}}^B_x) = \{3,4\}, \underset{s \in col({\bf{M}}_x,{\bf{F}}^{B}_x)} {\sum} mrk_q({\bf{F}}^{(s)}_x)=5.
\end{gather*}
Consider the completion of ${\bf{F}}^B_x$ given below. Observe that the last two rows are independent and span $\langle {\bf{F}}^B \rangle$. Note that this choice of ${\bf{F}}^{B}$ and ${\bf{M}}_x$ satisfy conditions $(i)$ and $(ii)$ given in the theorem.
\[{\bf{F}}^{B}=
\left(
\begin{array}{cccc} 
1 & 0 & 1 & 0\\
0 & 1 & 0 & 1\\
1 & 0 & 1 & 0\\
0 & 1 & 0 & 1\\
\end{array}
\right).
\]
We complete ${\bf{F}}^{E}_x$ as given in Lemma \ref{upbnd}. The encoding matrix ${\bf{\hat{F}}}^{E}$ obtained by this completion  is also given below. Observe that the codelength is $3+2=5$ as stated by the theorem.
\[{\bf{F}}^{E}=
\left(
\begin{array}{cc|ccc|cccc|cccc} 
1 & 0 & 0 & 0 & 0 & 1 & 0 & 1 & 0 & 0 & 0 & 0 & 0\\
0 & 1 & 0 & 0 & 0 & 0 & 1 & 0 & 1 & 0 & 0 & 0 & 0\\
\hline
0 & 0 & 1 & 1 & 0 & 0 & 0 & 0 & 0 & 1 & 0 & 1 & 0\\
0 & 0 & 0 & 1 & 1 & 0 & 0 & 0 & 0 & 0 & 1 & 0 & 1\\
0 & 0 & 1 & 0 & 1 & 0 & 0 & 0 & 0 & 1 & 1 & 1 & 1\\
\hline
1 & 0 & 0 & 0 & 0 & 1 & 0 & 1 & 0 & 0 & 0 & 0 & 0\\
0 & 1 & 0 & 0 & 0 & 0 & 1 & 0 & 1 & 0 & 0 & 0 & 0\\
1 & 0 & 0 & 0 & 0 & 1 & 0 & 1 & 0 & 0 & 0 & 0 & 0\\	
0 & 1 & 0 & 0 & 0 & 0 & 1 & 0 & 1 & 0 & 0 & 0 & 0\\	
\hline
0 & 0 & 1 & 1 & 0 & 0 & 0 & 0 & 0 & 1 & 1 & 0 & 0\\
0 & 0 & 0 & 1 & 1 & 0 & 0 & 0 & 0 & 0 & 1 & 0 & 1\\
0 & 0 & 0 & 0 & 0 & 0 & 0 & 0 & 0 & 0 & 0 & 1 & 1\\
0 & 0 & 1 & 1 & 0 & 0 & 0 & 0 & 0 & 1 & 1 & 0 & 0\\
0 & 0 & 0 & 1 & 1 & 0 & 0 & 0 & 0 & 0 & 1 & 0 & 1\\
\end{array}
\right).
\]
\[{\bf{\hat{F}}}^{E}=
\left(
\begin{array}{cc|ccc|cccc|cccc} 
1 & 0 & 0 & 0 & 0 & 1 & 0 & 1 & 0 & 0 & 0 & 0 & 0\\
0 & 1 & 0 & 0 & 0 & 0 & 1 & 0 & 1 & 0 & 0 & 0 & 0\\
\hline
0 & 0 & 1 & 1 & 0 & 0 & 0 & 0 & 0 & 1 & 1 & 0 & 0\\
0 & 0 & 0 & 1 & 1 & 0 & 0 & 0 & 0 & 0 & 1 & 0 & 1\\
0 & 0 & 0 & 0 & 0 & 0 & 0 & 0 & 0 & 0 & 0 & 1 & 1\\
\end{array}
\right).
\]
\end{exmp}



\section{An algorithm to obtain scalar linear codes for a special class of jointly extended problems}

In this section, we provide an algorithm (Algorithm $1$) to construct scalar linear codes for the special class of jointly extended problems identified in this paper (Section II). Any scalar linear code of every sub-problem  and that of the base problem are given in terms of their encoding matrices as inputs to the algorithm. An encoding matrix of the given code for the $i$th sub-problem with codelength $r_i$ is denoted as ${\bf{G}}^{(i)}$, $i \in [m_B]$ (of size $r_i \times m_i$). An encoding matrix of the  given code for the base problem with codelength $r_B$ is denoted as ${\bf{G}}^{B}$ (of size $r_B \times m_B$).
An associated decoding matrix of the given code of the base problem (denoted as ${\bf{D}}^B$) and its fitting matrix are also inputs to the algorithm. The algorithm provides a scalar linear code for the jointly extended problem in terms of an encoding matrix denoted as ${\bf{G}}^E$. The given codes of the sub-problems and the base problem need not be optimal. The resulting scalar linear code of the jointly extended problem need not be optimal, even when all the related codes (inputs to the algorithm) are optimal. We can obtain different scalar linear codes for the jointly extended problem by providing different sets of codes for the sub-problems and the base problem as inputs to the algorithm. 

\begin{algorithm}
	\caption{An algorithm to construct a scalar linear code matrix for any jointly extended problem $\mathcal{I}_E(\mathcal{I}_{B};(\mathcal{I}_{i})_{i \in [m_B]})$.}
	\label{algo1}
	\begin{algorithmic}[1]
		\item [\bf{Inputs}:] Full-rank encoding matrices $\{{\bf{G}}^{(i)}\}_{i \in [m_B]}$ of problems $(\mathcal{I}_i)_{i \in [m_B]}$ with respective sizes given by $\{r_i \times m_i\}_{i \in [m_B]}$, a full-rank $r_B \times m_B$ encoding matrix ${\bf{G}}^{B}$ of $\mathcal{I}_B$, an associated  $n_B \times r_B$ decoding matrix ${\bf{D}}^{B}$, and the fitting matrix ${\bf{F}}_x^{B}$ of the base problem.
		\item [\bf{Output}:] An encoding matrix ${\bf{G}}^{E}$ giving a code for $\mathcal{I}_E$.
		\item [\bf{Procedure}:]	
		\State Let $\sigma : [m_B] \rightarrow [m_B]$ be any permutation such that $r_{\sigma(1)} \geq r_{\sigma(2)} \geq \cdots \geq  r_{\sigma(m_B)}$, $\Psi = [r_B]$, and $t=1$. 
		\While{$\Psi \neq \phi$} 
		\State $\mathcal{U}^{(t)} \triangleq \{ u : ({\bf{F}}_{x}^{B})_{u,\sigma(t)} = 1$, $u \in [n_B] \}.$ 
		\State Let $\mathcal{U}^{(t)} =   \{u^{(t,1)}, u^{(t,2)}, \cdots ,u^{(t,|\mathcal{U}^{(t)}|)} \}$.
		\State Initialize $\mathcal{A}^{(t)} = \phi$, $\mathcal{B}^{(0,j)}=\phi$, for all $j \in [m_B]$.
		\For {$i = 1$ to $|\mathcal{U}^{(t)}|$}
		\State $\mathcal{A}^{(t)} \leftarrow \{k : {\bf{D}}^{B}_{u^{(t,i)},k}{\bf{G}}^{B}_{k,\sigma(t)} \neq 0, k \in \Psi\}  \cup  \mathcal{A}^{(t)}$.
		\EndFor   
		\State
		Let $\mathcal{A}^{(t)} = \{a^{(t,1)},\cdots,a^{(t,|\mathcal{A}^{(t)}|)}\}$.
		\State Initialize $\mathcal{B}^{(t,j)}=\phi$, for all $j \in [m_B]$.
		\State Let $\mathcal{B}^{(t-1)} = \{\mathcal{B}^{(t-1,1)},\cdots,\mathcal{B}^{(t-1,m_B)}\}$.
		\State {\bf{FILL}}$(\mathcal{A}^{(t)},t,\{\mathcal{B}^{(t')}\}_{t' \in [t-1]}, \{(r_{i},{\bf{G}}^{(i)})\}_{i \in [m_B]})$.
		\State Set $\Psi \leftarrow \Psi \setminus \mathcal{A}^{(t)}$.
		\For {$i = 1$ to $|\mathcal{U}^{(t)}|$}
		\State $\mathcal{V}^{(t,i)} \triangleq \{v: ({\bf{F}}^{B}_x)_{u^{(t,i)},v}=0\}$.
		\State Let $\mathcal{V}^{(t,i)} = \{v^{(t,i)}_1,\cdots,v^{(t,i)}_{|\mathcal{V}^{(t,i)}|}\}$.
		\For {$j = 1$ to $|\mathcal{V}^{(t,i)}|$}
		\State $\mathcal{Y}=\{k : {\bf{D}}^{B}_{u^{(t,i)},k}{\bf{G}}^{B}_{k,v^{(t,i)}_j} \neq 0, ~k \in \Psi\}$.
		\If {$r_{v^{(t,i)}_j} < r_{\sigma(t)}$}
		\State $\mathcal{B}^{(t,v^{(t,i)}_j)} \leftarrow  \mathcal{Y} \cup  \mathcal{B}^{(t,v^{(t,i)}_j)}$.
		\Else 
		\State $\mathcal{B}^{(t,\sigma(t))} \leftarrow  \mathcal{Y} \cup  \mathcal{B}^{(t,\sigma(t))}$.
		\EndIf
		\EndFor
		\EndFor
		\If {$t=m_B$ and $\Psi \neq \phi$}
		\State Let $\Psi = \{a_1, \cdots, a_{|\Psi|}\}$.
		\For {$i = 1$ to $|\Psi|$}
		\State $\hat{r}_{a_i}= max \{r_k : a_i \in \mathcal{B}^{(t',k)}, (t',k) \in [m_B]$ $~~~~~~~~~~~~~~\times [m_B]\}$. Fill $(a_i,j)$th  block matrix of ${\bf{G}}^{E}$, $~~~~~~~~~~~~~~$ with ${\bf{G}}^{B}_{a_i,j}{\bf{\hat{G}}}^{(a_i,j)}$, $\forall j \in [m_B]$, where  ${\bf{\hat{G}}}^{(a_i,j)}$ $~~~~~~~~~~~~~~$ is given by  \[ 
		{\bf{\hat{G}}}^{(a_i,j)}=
		\begin{cases}
		{\bf{G}}^{(j)}_{[[\hat{r}_{a_i}]]} & if \ \hat{r}_{a_i} < r_j,\\
		\left( \begin{array}{c} {\bf{G}}^{(j)}\\
		\hline
		{\bf{0}}_{(\hat{r}_{a_i}-r_{j}) \times m_j}\\
		\end{array}
		\right)        & \ otherwise.
		\end{cases}
		\]
		\EndFor
		\State $\Psi \leftarrow \phi$.
		\EndIf
		\State  $t \leftarrow t+1$.
		\EndWhile
		\item [\bf{Return}:] Encoding matrix  ${\bf{G}}^E$ of size $ (\underset{i \in [r_B]}{\Sigma} \hat{r}_i) \times (\underset{j \in [m_B]}{\Sigma} m_j)$.
	\end{algorithmic}
\end{algorithm}

The algorithm uses the given codes of the sub-problems according to that of the base problem to complete the fitting matrix of the jointly extended problem. We assume the encoding matrix ${\bf{G}}^E$ provided by the algorithm to be of the form as given in (\ref{extGmat}) (having $r_B$ block-rows and $m_B$ block-columns). Similarly, an associated decoding matrix ${\bf{D}}^E$ is assumed to be of the form as given in (\ref{extDmat}). The validity of these encoding and decoding matrices is proved in Theorem \ref{lemexplbase}. For any given $j \in [m_B]$ and  all $i \in [r_B]$, the matrices $\hat{{\bf{G}}}^{(i,j)}$ are obtained from ${\bf{G}}^{(j)}$ by either appending appropriate number of all-zero rows or deleting appropriate number of rows as given in the accompanying algorithm ${\bf{FILL}}$ (Algorithm $2$).
Similarly, for any given $i \in [n_B]$ and all $j \in [r_B]$, the matrices $\hat{{\bf{D}}}^{(i,j)}$ are obtained from ${\bf{D}}^{(i)}$, by either appending all-zero columns or deleting appropriate number of columns as given in the proof of Theorem \ref{lemexplbase}. The accompanying algorithm named as {\bf{FILL}} progressively fills the $r_B$ block-rows of ${\bf{G}}^E$. The number of rows in each block-row is also decided and given as $\hat{r}_i$, $i \in [r_B]$, in {\bf{FILL}}.

\begin{algorithm}
	\caption{The algorithm to fill some block-rows of ${\bf{G}}^E$.}
	\begin{algorithmic}[1]
		\item [\bf{Inputs}:] $\mathcal{A}^{(t)}$, $t$,  $\{\mathcal{B}^{(t')}\}_{t' \in [t-1] }$, $\{(r_{i},{\bf{G}}^{(i)})\}_{i \in [m_B]}$.
		\item [\bf{Outputs}:] Filled block-rows of ${\bf{G}}^{E}$ and the number or rows  $~~~~~~~~$ in each block-row with row indices in $\mathcal{A}^{(t)}$.  
		\item [\bf{Procedure}:]  {\bf{FILL}}$(\mathcal{A}^{(t)},t,\{\mathcal{B}^{(t')}\}_{t' \in [t-1]}, \{(r_{i},{\bf{G}}^{(i)})\}_{i \in [m_B]})$.	
		\For {$i = 1$ to $|\mathcal{A}^{(t)}|$}
		\State $\hat{r}_{a^{(t,i)}}= max(r_{\sigma(t)}, max \{r_k : a^{(t,i)} \in \mathcal{B}^{(t',k)}, (t',k)$ $~~~~~~~~~~~~~~~~\in [t-1] \times [m_B]\})$.
		\State Fill the $(a^{(t,i)},j)$th block matrix of ${\bf{G}}^{E}$, $\forall j \in [m_B]$, $~~~~$ with ${\bf{G}}^{B}_{a^{(t,i)},j}{\bf{\hat{G}}}^{(a^{(t,i)},j)}$, where ${\bf{\hat{G}}}^{(a^{(t,i)},j)}$ is given by $~~~~~ 
		{\bf{\hat{G}}}^{(a^{(t,i)},j)}=
		\begin{cases}
		{\bf{G}}^{(j)}_{[[\hat{r}_{a^{(t,i)}}]]} & if \ \hat{r}_{a^{(t,i)}} < r_j,\\
		\left( \begin{array}{c} {\bf{G}}^{(j)}\\
		\hline
		{\bf{0}}_{(\hat{r}_{a^{(t,i)}}-r_{j}) \times m_j}\\
		\end{array}
		\right)        & \ otherwise.
		\end{cases}
		$
		\EndFor
		\item [\bf{Return}:] Block-rows of  ${\bf{G}}^E$ and the number of rows in each block-row with row indices in $\mathcal{A}^{(t)}$, given by $\{\hat{r}_{l}\}_{l \in \mathcal{A}^{(t)}}$.
	\end{algorithmic}
	\label{subr1}
\end{algorithm}

We first choose a permutation $\sigma$ of the set $[m_B]$. The permutation $\sigma$ orders the codelengths of the given codes of all the sub-problems in any non-increasing order as given in Line $1$. The set $\Psi$ (initialized in Line $1$) consists of indices of the block-rows of ${\bf{G}}^E$ that are not filled until the current iteration of the while loop. The while loop iterates until all the block-rows of ${\bf{G}}^E$ are filled. Variable `$t$' tracks the iteration number.  Any $k$th row of ${\bf{G}}^B$ (equivalently $k$th block-row of ${\bf{G}}^E$), where $k \in [r_B]$, is said to contribute to the completion of the $(i,j)$th entry of ${\bf{F}}_x^B$ (equivalently $(i,j)$th   block matrix of ${\bf{F}}_x^E$) for $i \in [n_B]$, $j \in [m_B]$, if ${\bf{D}}^B_{i,k}{\bf{G}}^B_{k,j}\neq 0$. In $t$th iteration of the while loop, we fill the block-rows of ${\bf{G}}^E$ that contribute to the completion of all the occurrences of ${\bf{F}}_x^{(\sigma(t))}$ (in ${\bf{F}}_x^E$), that were not completed in previous iterations. This is explained further in the following.   

In $t$th iteration, the set $\mathcal{U}^{(t)}$ consists of the indices of all rows of ${\bf{F}}_x^B$, which have a $1$ in $\sigma(t)$th column (Line $4$). Note that this is same as the set of indices of all the block-rows of ${\bf{F}}_x^E$, which have ${\bf{F}}_x^{(\sigma(t))}$ in $\sigma(t)$th block-column. The set $\mathcal{A}^{(t)}$ consists of indices of all the rows of  ${\bf{G}}^B$ that are present in $\Psi$, and contribute to the completion of some or all entries in $\sigma(t)$th column of ${\bf{F}}_x^B$ consisting of $1$'s (Line $7$). Hence, $\mathcal{A}^{(t)}$ consists of indices of all the block-rows of  ${\bf{G}}^E$ that are present in $\Psi$ (that is, they are not filled in any previous iterations), and contribute to completing some or all occurrences of ${\bf{F}}_x^{(\sigma(t))}$ in $\sigma(t)$th block-column of ${\bf{F}}_x^E$. This is made clear in the proof of Theorem \ref{lemexplbase} and follows from taking the product   ${\bf{D}}^E$${\bf{G}}^E$. Note that some occurrences of  ${\bf{F}}_x^{(\sigma(t))}$ in $\sigma(t)$th block-column of ${\bf{F}}_x^E$ might be completed in previous iterations, as a result of completing the fitting matrices of sub-problems in ${\bf{F}}_x^E$ with larger given codelengths. The remaining occurrences of ${\bf{F}}_x^{(\sigma(t))}$ (if any are present) are completed in $t$th iteration. Thus, by the end of $t$th iteration, we fill all the the block-rows of ${\bf{G}}^E$ that contribute to the completion of all the occurrences of ${\bf{F}}_x^{(\sigma(t))}$ in ${\bf{F}}_x^{E}$. 

Note that in $t$th iteration, the accompanying algorithm ${\bf{FILL}}$ fills the block-rows of ${\bf{G}}^E$, which contributes to completing those occurrences of ${\bf{F}}_x^{(\sigma(t))}$ in ${\bf{F}}_x^{E}$, which were not completed by the block-rows filled in the previous iterations. 
Observe that the index of any row of ${\bf{G}}^B$ present in $[r_B] \setminus \Psi$ and contributing to completing any entry of ${\bf{F}}_x^B$ is not included in $\mathcal{A}^{(t)}$. These indices need not be taken into consideration, as the corresponding block-rows have already  contributed in completing the fitting matrices of sub-problems in ${\bf{F}}_x^E$, with larger given codelengths than that of the sub-problem addressed in the current iteration.  

The algorithm ${\bf{FILL}}$ fills the block-rows of ${\bf{G}}^E$ with indices given by $\mathcal{A}^{(t)}$, after deciding the number of rows required for each of these block-rows. In the $t$th iteration, the number of rows required for each block-row is decided by the sets $\{\mathcal{B}^{(t')}\}_{t' \in [t-1]}$. These sets are filled in previous iterations as given from Line $14$ to Line $25$ of Algorithm $1$.  Note that they are populated after updating the set $\Psi$.  For every row index $u^{(t,i)}$ present in $\mathcal{U}^{(t)}$, $i \in [|\mathcal{U}^{(t)}|]$, the set  $\mathcal{V}^{(t,i)}$ consists of the indices of block-columns containing ${\bf{0}}$ matrices present in $u^{(t,i)}$th block-row of ${\bf{F}}_x^E$. Note that the filling of block-rows of ${\bf{G}}^E$ must also satisfy the completion of these block matrices by  ${\bf{0}}$ matrices. For every index of the block-column present in $\mathcal{V}^{(t,i)}$, $i \in [|\mathcal{U}^{(t)}|]$, the set $\mathcal{Y}$ consists of the indices of block-rows of ${\bf{G}}^E$, that contribute to completing the ${\bf{0}}$ matrix in the given block-column and $u^{(t,i)}$th block-row. The set $\mathcal{B}^{(t,j)}$ consists of indices of block-rows of ${\bf{G}}^E$ that contribute to completing the   ${\bf{0}}$ matrices present in $j$th block-column of ${\bf{F}}_x^E$, for $j \in [m_B]$. This set of indices is related to completing ${\bf{0}}$ matrices in the block-rows with indices in $\mathcal{A}^{(t)}$. Note that if $r_j \geq r_{\sigma(t)}$, then $\mathcal{B}^{(t,\sigma(t))}$ is updated instead of $\mathcal{B}^{(t,j)}$. This point will be made clear in the proof of Theorem \ref{lemexplbase}. The Lines $26$ to $30$ describe the filling of any remaining block-rows of ${\bf{G}}^E$, after $m_B$ iterations of the while loop. 

If the index of any block-row of ${\bf{G}}^E$ contributing to complete any occurrence of  ${\bf{F}}_x^{(\sigma(t))}$ is not present in any set $\mathcal{B}^{(t',k)}$, for  $t' \in [t-1]$, and $k \in [m_B]$, then the block-row is assigned $r_{\sigma(t)}$ rows. This implies that the desired block-row did not contribute to complete any block matrix of ${\bf{F}}_x^E$ in the previous iterations. However, if the index of any block-row of ${\bf{G}}^E$ contributing to complete any occurrence of  ${\bf{F}}_x^{(\sigma(t))}$ is present in some set $\mathcal{B}^{(t',k)}$, for  $t' \in [t-1]$, and $k \in [m_B]$, then the block-row is assigned the number of rows as given in Line $2$ of ${\bf{FILL}}$. 

We illustrate Algorithm $1$ with an example and then provide a proof to show its correctness.

\begin{exmp}
	Consider $m_B=n_B=5$. The fitting matrices ${\bf{F}}_x^{B}$ and $({\bf{F}}_x^{(i)})_{i \in [m_B]}$ are given below.
	\[
	{\bf{F}}_x^{B}=
	\left(
	\begin{array}{ccccc} 
	1  & x & x & 0 & 0 \\
	0  & 1 & x & x & 0 \\
	0  & 0 & 1 & x & x \\
	x  & 0 & 0 & 1 & x \\
	x  & x & 0 & 0 & 1 \\
	\end{array}
	\right),
	{\bf{F}}_x^{(1)}=
	\left(
	\begin{array}{cccc} 
	1 & x & 0 & 0\\
	0 & 1 & x & 0\\
	0 & 0 & 1 & x\\
	x & 0 & 0 & 1\\
	\end{array}
	\right),
	\]	
	\[
	{\bf{F}}_x^{(2)}=
	\left(
	\begin{array}{ccc} 
	1 & x & 0\\
	0 & 1 & x\\
	x & 0 & 1\\
	\end{array}
	\right),
	{\bf{F}}_x^{(3)}=
	\left(
	\begin{array}{ccc} 
	1 & 0 & x \\
	x & 1 & 0 \\
	x & x & 1 \\
	\end{array}
	\right),
	\]
	\[
	{\bf{F}}_x^{(4)}=	
	\left(
	\begin{array}{cc} 
	1 & x \\
	x & 1 \\
	\end{array}
	\right),
	{\bf{F}}_x^{(5)}=	
	\left(
	\begin{array}{c} 
	1  \\
	\end{array}
	\right).
	\]
${\bf{F}}_x^{E}$ is given below without vertical partitions for reference.
	\[
	\left(
	\begin{array}{ccccccccccccc} 
	1 & x & 0 & 0 & x & x & x & x & x & x & 0 & 0 & 0 \\
	0 & 1 & x & 0 & x & x & x & x & x & x & 0 & 0 & 0 \\
	0 & 0 & 1 & x & x & x & x & x & x & x & 0 & 0 & 0 \\
	x & 0 & 0 & 1 & x & x & x & x & x & x & 0 & 0 & 0 \\
	\hline
	0 & 0 & 0 & 0 & 1 & x & 0 & x & x & x & x & x & 0 \\
	0 & 0 & 0 & 0 & 0 & 1 & x & x & x & x & x & x & 0\\
	0 & 0 & 0 & 0 & x & 0 & 1 & x & x & x & x & x & 0\\
	\hline
	0 & 0 & 0 & 0 & 0 & 0 & 0 & 1 & 0 & x & x & x & x\\
	0 & 0 & 0 & 0 & 0 & 0 & 0 & x & 1 & 0 & x & x & x\\
	0 & 0 & 0 & 0 & 0 & 0 & 0 & 0 & x & 1 & x & x & x\\
	\hline
	x & x & x & x & 0 & 0 & 0 & 0 & 0 & 0 & 1 & x & x\\
	x & x & x & x & 0 & 0 & 0 & 0 & 0 & 0 & x & 1 & x\\
	\hline 
	x & x & x & x & x & x & x & 0 & 0 & 0 & 0 & 0 & 1\\
	\end{array}
	\right).
	\]  
	Let us consider some optimal encoding matrices of the related problems and an associated decoding matrix of the base problem as given below. These matrices are taken as inputs to Algorithm $1$. Note that $r_1=3,r_2=r_3=2,r_4=r_5=1$, and $r_B=3$. Let $\sigma$ be the identity permutation. That is, $\sigma(i)=i$, $\forall i \in \{1,2,3,4,5\}$. Initially, $\Psi = \{1,2,3\}$.
		\[
		{\bf{G}}^{B}=
		\left(
		\begin{array}{ccccc} 
		0  & 1 & 1 & 1 & 0 \\
		0  & 0 & 1 & 1 & 1 \\
		1  & 0 & 0 & 1 & 0 \\
		\end{array}
		\right),
		{\bf{D}}^{B}=
		\left(
		\begin{array}{ccc} 
		1 & 0 & 1\\
	    1 & 0 & 0\\
		0 & 1 & 0\\
		0 & 0 & 1\\
		1 & 1 & 0\\	
		\end{array}
		\right),
		\]	
		\[
		{\bf{G}}^{(1)}=
		\left(
		\begin{array}{cccc} 
		1 & 1 & 0 & 0\\
		0 & 1 & 1 & 0\\
		0 & 0 & 1 & 1\\
		\end{array}
		\right),
		{\bf{G}}^{(2)}={\bf{G}}^{(3)}=
		\left(
		\begin{array}{ccc} 
		1 & 1 & 0\\
		0 & 1 & 1\\
		\end{array}
	   \right),
		\]
				\[
				{\bf{G}}^{(4)}=
				\left(
				\begin{array}{cc} 
				1 & 1 \\
				\end{array}
				\right),
				{\bf{G}}^{(5)}=
					\left(
					\begin{array}{c} 
					1 \\
					\end{array}
					\right).
					\]
	For the first iteration, $\mathcal{U}^{(1)}=1$, $\mathcal{A}^{(1)}=3$, and hence $\hat{r}_{3}=r_1$ according to ${\bf{FILL}}$. Hence, the third block-row of ${\bf{G}}^{E}$ is given as below. Now $\Psi=\{1,2\}$, and $\mathcal{V}^{(1,1)}=\{4,5\}$. Hence, we get $\mathcal{B}^{(1,4)}=\{1,3\}$, and $\mathcal{B}^{(1,1)}=\mathcal{B}^{(1,2)}=\mathcal{B}^{(1,3)}=\mathcal{B}^{(1,5)}=\Phi$.
				\[
				\left(
				\begin{array}{ccccccccccccc} 
				1 & 1 & 0 & 0 & 0 & 0 & 0 & 0 & 0 & 0 & 1 & 1 & 0 \\
				0 & 1 & 1 & 0 & 0 & 0 & 0 & 0 & 0 & 0 & 0 & 0 & 0 \\
				0 & 0 & 1 & 1 & 0 & 0 & 0 & 0 & 0 & 0 & 0 & 0 & 0 \\
				\end{array}
				\right).
				\]		
	For the second iteration, $\mathcal{U}^{(2)}=2$, $\mathcal{A}^{(2)}=1$, and hence $\hat{r}_{1}=r_2$ according to ${\bf{FILL}}$. Hence, the first block-row of ${\bf{G}}^{E}$ is given as below. Now $\Psi=\{2\}$, and $\mathcal{V}^{(2,1)}=\{1,5\}$. Hence,  $\mathcal{B}^{(2,1)}=\mathcal{B}^{(2,2)}=\mathcal{B}^{(2,3)}=\mathcal{B}^{(2,4)}=\mathcal{B}^{(2,5)}=\Phi$.
	\[
	\left(
	\begin{array}{ccccccccccccc} 
	0 & 0 & 0 & 0 & 1 & 1 & 0 & 1 & 1 & 0 & 1 & 1 & 0 \\
	0 & 0 & 0 & 0 & 0 & 1 & 1 & 0 & 1 & 1 & 0 & 0 & 0 \\
	\end{array}
	\right).
	\]	
		For the third iteration, $\mathcal{U}^{(3)}=3$, $\mathcal{A}^{(3)}=2$, and hence $\hat{r}_{2}=r_2$ according to ${\bf{FILL}}$. Hence, the second block-row of ${\bf{G}}^{E}$ is given as below. Now $\Psi=\Phi$, and the algorithm terminates.
		\[
		\left(
		\begin{array}{ccccccccccccc} 
		0 & 0 & 0 & 0 & 0 & 0 & 0 & 1 & 1 & 0 & 1 & 1 & 1 \\
		0 & 0 & 0 & 0 & 0 & 0 & 0 & 0 & 1 & 1 & 0 & 0 & 0 \\
		\end{array}
		\right).
		\]	
     Hence, the overall encoding matrix ${\bf{G}}^{E}$ is given as below.
     		\[{\bf{G}}^{E} =
     		\left(
     		\begin{array}{cccc|ccc|ccc|cc|c} 
     		0 & 0 & 0 & 0 & 1 & 1 & 0 & 1 & 1 & 0 & 1 & 1 & 0 \\
     		0 & 0 & 0 & 0 & 0 & 1 & 1 & 0 & 1 & 1 & 0 & 0 & 0 \\
     		\hline
   			0 & 0 & 0 & 0 & 0 & 0 & 0 & 1 & 1 & 0 & 1 & 1 & 1 \\
     		0 & 0 & 0 & 0 & 0 & 0 & 0 & 0 & 1 & 1 & 0 & 0 & 0 \\
     		\hline
     		1 & 1 & 0 & 0 & 0 & 0 & 0 & 0 & 0 & 0 & 1 & 1 & 0 \\
     		0 & 1 & 1 & 0 & 0 & 0 & 0 & 0 & 0 & 0 & 0 & 0 & 0 \\
     		0 & 0 & 1 & 1 & 0 & 0 & 0 & 0 & 0 & 0 & 0 & 0 & 0 \\
     		\end{array}
     		\right).
     		\]
     It can be easily verified that all the receivers are able to decode their demands from the code obtained using ${\bf{G}}^{E}$.		
	\label{exmp8}
\end{exmp}

We use the following lemma which was stated in \cite{PK}, as a necessary and sufficient condition for a given  matrix to be an encoding matrix for the given index coding problem.

\begin{lem}[Lemma 1, \cite{PK}]
	For an index coding problem $\mathcal{I}$ (groupcast or single unicast) with $n \times m$ fitting matrix ${\bf{F}}_x$, a matrix ${\bf{G}} \in \mathbb{F}^{r \times m}_q$ is an encoding matrix iff there exists a matrix ${\bf{D}} \in \mathbb{F}^{n \times r}_q$ such that ${\bf{D}}{\bf{G}}$ completes ${\bf{F}}_x$, i.e. ${\bf{D}}{\bf{G}} \approx  {\bf{F}}_x$.\\
	\label{dgmatrix}
\end{lem} 

It can be easily observed that the indices of the non-zero entries of  the $i$th row of ${\bf{D}}$ are same as the indices of the code symbols that must be used by the $i$th receiver to decode its demanded message, $i \in [n]$. Hence, we call the matrix ${\bf{D}}$ as an associated decoding matrix for the index code given by the encoding matrix ${\bf{G}}$. Note that for a given encoding matrix, there need not exist a unique  associated decoding matrix, but many associated decoding matrices can exist. We now use Lemma \ref{dgmatrix} to prove the correctness of Algorithm $1$.

\begin{equation}
{\bf{G}}^{E}=
\left(
\begin{array}{c|c|c|c} 
{\bf{G}}^{B}_{1,1}{\bf{\hat{G}}}^{(1,1)} & \cdots & \cdots & {\bf{G}}^{B}_{1,m_B}{\bf{\hat{G}}}^{(1,m_B)} \\
\hline
{\bf{G}}^{B}_{2,1}{\bf{\hat{G}}}^{(2,1)}  & \cdots & \cdots & {\bf{G}}^{B}_{2,m_B}{\bf{\hat{G}}}^{(2,m_B)} \\
\hline
\cdots & \ldots & \ldots & \cdots \\
\hline
{\bf{G}}^{B}_{r_B,1}{\bf{\hat{G}}}^{(r_B,1)} & \cdots & \cdots & {\bf{G}}^{B}_{r_B,m_B}{\bf{\hat{G}}}^{(r_B,m_B)} \\
\end{array}
\right).
\label{extGmat}
\end{equation}
\begin{equation}
{\bf{D}}^{E}=
\left(
\begin{array}{c|c|c|c} 
{\bf{D}}^{B}_{1,1}{\bf{\hat{D}}}^{(1,1)} & \cdots & \cdots & {\bf{D}}^{B}_{1,r_B}{\bf{\hat{D}}}^{(1,r_B)} \\
\hline
{\bf{D}}^{B}_{2,1}{\bf{\hat{D}}}^{(2,1)}  & \cdots & \cdots & {\bf{D}}^{B}_{2,r_B}{\bf{\hat{D}}}^{(2,r_B)} \\
\hline
\cdots & \ldots & \ldots & \cdots \\
\hline
{\bf{D}}^{B}_{n_B,1}{\bf{\hat{D}}}^{(n_B,1)} & \cdots & \cdots & {\bf{D}}^{B}_{n_B,r_B}{\bf{\hat{D}}}^{(n_B,r_B)} \\
\end{array}
\right).
\label{extDmat}
\end{equation}

\begin{thm}
	For any problem $\mathcal{I}_E(\mathcal{I}_B;(\mathcal{I}_i)_{i \in [m_B]})$, the matrix ${\bf{G}}^{E}$ obtained using Algorithm \ref{algo1} is a valid encoding matrix.
	\label{lemexplbase}
\end{thm}
\begin{proof}
	To prove this theorem, we construct a matrix ${\bf{D}}^{E}$ using $n_B \times r_B$ block matrices as in (\ref{extDmat}). It is obtained by using corresponding decoding matrices $\{{\bf{D}}^{(i)}\}_{i \in [m_B]}$ of the encoding matrices $\{{\bf{G}}^{(i)}\}_{i \in [m_B]}$, and the decoding matrix ${\bf{D}}^{B}$ employed in Algorithm \ref{algo1}. We then show that ${\bf{D}}^{E}$ and ${\bf{G}}^{E}$ are a valid pair of decoding and encoding matrices using Lemma \ref{dgmatrix}. The $i$th decoding matrix ${\bf{D}}^{(i)}$ is of size $n_i \times r_i$, $i \in [m_B]$.	
	
	The $(i,j)$th block matrix of ${\bf{D}}^{E}$ be given by  ${\bf{D}}^{B}_{i,j}{\bf{\hat{D}}}^{(i,j)}$, $i \in [n_B]$, $j \in [r_B]$, where ${\bf{\hat{D}}}^{(i,j)}$ is given by 
	\[ 
	{\bf{\hat{D}}}^{(i,j)}=
	\begin{cases}
	({\bf{D}}^{(f(i))})^{[[\hat{r}_j]]} & if \ r_{f(i)} > \hat{r}_j,\\
	({\bf{D}}^{(f(i))} | {\bf{0}}_{n_{f(i)} \times (\hat{r}_{j}-r_{f(i)})})        & \ otherwise.
	\end{cases}
	\]	
	where $f(i)$ is the index of the message demanded by the $i$th receiver and $\hat{r}_j$ is the number of rows in the $j$th block-row of ${\bf{G}}^{E}$ as assigned in ${\bf{FILL}}$. Note that when $r_{f(i)} > \hat{r}_j$, only the first $\hat{r}_j$ columns of ${\bf{D}}^{(f(i))}$ are taken.
	
	We analyze any $i$th block-row of  ${\bf{D}}^{E}{\bf{G}}^{E}$, $i \in [n_B]$, and prove that it completes $i$th block-row of  ${\bf{F}}_x^{E}$. ${\bf{D}}^{E}{\bf{G}}^{E}$ consists of $n_B \times m_B$ block matrices. Any $i$th block-row of ${\bf{D}}^{E}{\bf{G}}^{E}$, $i \in [n_B]$, is given by 
	$({\bf{D}}^{B}_{i,1}{\bf{\hat{D}}}^{(i,1)}|\cdots|\cdots|{\bf{D}}^{B}_{i,r_B}{\bf{\hat{D}}}^{(i,r_B)}){\bf{G}}^{E}$. We first verify that $(i,f(i))$th block matrix of ${\bf{D}}^{E}{\bf{G}}^{E}$ completes ${\bf{F}}_x^{(f(i))}$, and is thus given by  $\underset{j \in [r_B]} {\Sigma} {\bf{D}}^{B}_{i,j}{\bf{\hat{D}}}^{(i,j)}{\bf{G}}^{B}_{j,f(i)}{\bf{\hat{G}}}^{(j,f(i))}$, for any  $i \in [n_B]$. From the construction given in Algorithm $1$ and its explanation, note that the matrices  ${\bf{\hat{G}}}^{(j,f(i))}$ are not obtained by deleting any rows of  ${\bf{G}}^{(f(i))}$, for those $j \in [r_B]$ which contribute to completing ${\bf{F}}^{(f(i))}_x$, $i \in [n_B]$. Hence, for such $j \in [r_B]$, note that ${\bf{\hat{D}}}^{(i,j)}{\bf{\hat{G}}}^{(j,f(i))}={\bf{D}}^{(f(i))}{\bf{G}}^{(f(i))} \approx \mathbb{F}_x^{(f(i))}$. Hence, we have
	$\underset{j \in [r_B]} {\Sigma} {\bf{D}}^{B}_{i,j}{\bf{\hat{D}}}^{(i,j)}{\bf{G}}^{B}_{j,f(i)}{\bf{\hat{G}}}^{(j,f(i))} = {\bf{D}}^{(f(i))}{\bf{G}}^{(f(i))} (\underset{j \in [r_B]} {\Sigma} {\bf{D}}^{B}_{i,j}{\bf{G}}^{B}_{j,f(i)}) = {\bf{D}}^{(f(i))}{\bf{G}}^{(f(i))} \times (1) \approx \mathbb{F}_x^{(f(i))}.$ The last equality follows as ${\bf{D}}^{B}$ is a decoding matrix of ${\bf{G}}^{B}$, and then using Lemma \ref{dgmatrix}.
	
	Observe that there is no need to verify the completion of  $(i,j)$th block matrix of ${\bf{F}}_x^{E}$, if it consists of only $x$'s, for all $i \in [n_B], j \in [m_B]$. Considering the $i$th block-row of  ${\bf{D}}^{E}{\bf{G}}^{E}$, $i \in [n_B]$, we next verify that its $j$th block matrix is a zero matrix of appropriate dimensions (that is $(i,j)$th block matrix of $\mathbb{F}_x^{E}$), for $j \in [m_B]$ such that $(\mathbb{F}^{B}_x)_{i,j}=0$. This is given by $\underset{k \in [r_B]} {\Sigma} {\bf{D}}^{B}_{i,k}{\bf{\hat{D}}}^{(i,k)}{\bf{G}}^{B}_{k,j}{\bf{\hat{G}}}^{(k,j)}$. From the construction of Algorithm $1$, note that there is a possibility that ${\bf{\hat{G}}}^{(k,j)}$ is obtained by deleting the last $r_j-r_{f(i)}$ rows from  ${\bf{G}}^{(j)}$. When this happens, note that all matrices ${\bf{\hat{D}}}^{(i,k)}$, $k \in [r_B]$, only have $r_{f(i)}$ non-zero rows. Hence the matrix ${\bf{\hat{D}}}^{(i,k)}{\bf{\hat{G}}}^{(k,j)}$ is the same, say ${\bf{B}}$ for all values of $k$. Therefore, $\underset{k \in [r_B]} {\Sigma} {\bf{D}}^{B}_{i,k}{\bf{\hat{D}}}^{(i,k)}{\bf{G}}^{B}_{k,j}{\bf{\hat{G}}}^{(k,j)}= {\bf{B}}\underset{k \in [r_B]} {\Sigma} {\bf{D}}^{B}_{i,k}{\bf{G}}^{B}_{k,j}=0\times{\bf{B}}={\bf{0}}$. The remaining case is dealt on similar lines, and we obtain an all-zero matrix. Hence the result.

\end{proof}

\section{Optimality of the constructed codes}

In this section, we provide a necessary condition for the scalar linear optimality of the code constructed using Algorithm $1$, when the codes of the sub-problems and the base problem are also scalar linear optimal. The base problem is assumed to be an SUICP.

An SUICP with $m$ messages is called a cycle $C$ if $\mathcal{K}_i=x_{i+1}$, for $i \in [m-1]$, and $\mathcal{K}_m =x_1$. Any cycle consisting of $m$ messages can save at most one transmission compared to the naive transmission of all messages using the code $\mathcal{C}=\{x_1+x_2, x_2+x_3, \cdots, x_{m-1}+x_m\}$ \cite{BK}, which is said to be a cyclic code \cite{SUOH}. 

 We now have the following theorem which makes use of Lemma \ref{lowbnd} to prove the optimality of the constructed code.
 
\begin{thm}
 	For any problem $\mathcal{I}_E(\mathcal{I}_B;(\mathcal{I}_i)_{i \in [m_B]})$, with the base problem being a cycle, the matrix ${\bf{G}}^{E}$ obtained using Algorithm \ref{algo1} gives an optimal scalar linear code, when all the given component codes are optimal scalar linear, and the given code for the base problem is the cyclic code.
 	\label{thm2}
\end{thm}
\begin{proof}
	Without loss of generality, we can assume that the cycle (the base problem) of $m_B$ messages is such that   $\mathcal{K}_i=x_{i+1}$, for $i \in [m_B-1]$, and $\mathcal{K}_{m_B} =x_1$. It can be easily verified that the following $m_B \times (m_B-1)$ matrix ${\bf{D}}^B$ is a valid decoding matrix for the code given by $\mathcal{C}=\{x_1+x_2, x_2+x_3, \cdots, x_{m-1}+x_m\}$. 
	\[{\bf{D}}^{B}=
	\left(
	\begin{array}{ccccc} 
	1 & 0 & \cdots & \cdots & 0 \\
	0 & 1 & 0 & \cdots & 0 \\
	\vdots & \ddots & \ddots & \ddots & \vdots\\
	0 & \ddots & \ddots & 1 & 0\\	
	0 & 0 & 0 & 0 & 1\\
	\hline
	(-1)^{1-1} & (-1)^{2-1} & \cdots & \cdots &  (-1)^{m_B-2}\\
	\end{array}
	\right).
	\]
    The length of the constructed code for the jointly extended problem can also be easily verified to be $\underset{i \in [m_B]} {\sum} mrk_q({\bf{F}}^{(i)}_x)-min\{mrk_q({\bf{F}}^{(i)}_x), i \in [m_B] \}$. We prove that this is a lower bound on $mrk_q({\bf{F}}^{E}_x)$ using Lemma \ref{lowbnd}. Let $ind_{min}=argmin\{mrk_q({\bf{F}}^{(i)}_x), i \in [m_B]\}$. Consider the submatrix ${\bf{M}}_x$ of ${\bf{F}}^{B}_x$ with $col({\bf{M}}_x,{\bf{F}}^{B}_x)=col({\bf{M}}_x,{\bf{F}}^{B}_x)=[m_B] \setminus ind_{min}$. This is an upper-triangulable submatrix of ${\bf{F}}^{B}_x$. Hence, the lower bound given by Lemma \ref{lowbnd}  is equal to  $\underset{i \in [m_B]} {\sum} mrk_q({\bf{F}}^{(i)}_x)-min\{mrk_q({\bf{F}}^{(i)}_x), i \in [m_B] \}$. Hence the result.
\end{proof} 

We illustrate the theorem with an example.

\begin{exmp} For the jointly extended problem with the fitting matrix given below, we observe that the base problem has an optimal scalar code of length $2$ given by the cyclic code $\mathcal{C}^B=(x_B^1+x_B^2,x_B^2+x_B^3)$. 
	\[
	{\bf{F}}_x^{E}=
	\left(
	\begin{array}{ccc|cc|cc} 
	1 & x & 0 & x & x & 0 & 0 \\
	0 & 1 & x & x & x & 0 & 0 \\
	x & 0 & 1 & x & x & 0 & 0 \\
	\hline
	0 & 0 & 0 & 1 & x & x & x \\
	0 & 0 & 0 & x & 1 & x & x \\
	\hline
	x & x & x & 0 & 0 & 1 & x \\
	x & x & x & 0 & 0 & x & 1 \\
	\end{array}
	\right).
	\]  
	Considering the following optimal component codes given by  $\mathcal{C}^1=(x_1+x_2,x_2+x_3)$, $\mathcal{C}^2=(x_1+x_2)$, and $\mathcal{C}^3=(x_1+x_2)$, we obtain the following optimal code after running the Algorithm $1$ : $\mathcal{C}^E=(x_1+x_2+x_4+x_5,x_2+x_3,x_4+x_5+x_6+x_7)$. It can be easily verified that all the receivers are able to decode their demands. Note also that the lower bound given by Lemma \ref{lowbnd} is also $3$, and hence this is  an optimal code.
\end{exmp}



\section{Conclusion and Future Work}   
A class of joint extensions of a finite number of groupcast ICPs is identified in this paper. A lower bound on the minrank of the jointly extended problem, a code construction algorithm (not necessarily optimal) for the special class of jointly extended problems, and a set of necessary conditions for their optimality are given in terms of those of the base problem and all the sub-problems.

Finding more classes of jointly extended problems, with optimal results being expressible in terms of their component problems is an interesting direction for future work. 

\section*{Acknowledgment}
This work was supported partly by the Science and Engineering Research Board (SERB) of Department of Science and Technology (DST), Government of India, through J.C. Bose National Fellowship to B. S. Rajan.


\begin{thebibliography}{10}
	
\bibitem{BK}
Y.~Birk, and T.~Kol, ``Coding on demand by an informed source ({ISCOD}) for
efficient broadcast of different supplemental data to caching clients,''
\emph{IEEE Transactions on Information Theory}, vol.~52, no.~6, pp.
2825--2830, 2006.

\bibitem{jafar}
S.~A. Jafar, ``Interference alignment - a new look at signal dimensions in a
 communication network,'' \emph{Foundations and Trends in Communications and Information Theory}, vol.~7, pp. 1--136, 2011.
 
\bibitem{SUOH}
L.~Ong, C.~K. Ho, and F.~Lim, ``The single-uniprior index-coding problem: The
single-sender case and the multi-sender extension,'' \emph{IEEE Transactions
	on Information Theory}, vol.~62, no.~6, pp. 3165--3182, 2016.

\bibitem{LNSG}
S. Ghosh, and L. Natarajan, ``Linear codes for broadcasting with noisy side information,'' \emph{arXiv preprint arXiv:1801.02868v1 [cs.IT]}, 9 Jan, 2018.


\bibitem{MV}
M.~B. Vaddi and B.~S. Rajan, ``{Optimal vector linear index codes for some symmetric side information problems}," in \textit{Proc. IEEE International Symposium on Information Theory}, Barcelona, Spain,  2016, pp.125-129.

\bibitem{LS}
E. Lubetzky, and U. Stav, ``Nonlinear index coding outperforming the
linear optimum," \textit{IEEE Transactions on Information Theory}, vol.~55, no.~8, pp.
3544--3551, 2009.

\bibitem{DSC1}
S.~H. Dau, V. Skachek, and Y.~M. Chee, ``Optimal index codes with near-extreme rates,"
\emph{IEEE Transactions on Information Theory}, vol. 60, no. 3, pp. 1515--1527, 2014.

\bibitem{DIC} 
M.~J. Neely, A.~S. Tehrani, and Z. Zhang, ``Dynamic index coding for wireless
	broadcast networks," \textit{IEEE Transactions on Information Theory}, vol. 59, no. 11, pp. 7525--7540, 2013.

\bibitem{ALSW}
N. Alon, A. Hasidim, E. Lubetzky, U. Stav, 
and A. Weinstein, ``Broadcasting with side information," \emph{arXiv preprint arXiv:0806.3245v1 [cs.IT]}, 19 Jun, 2008.


\bibitem{SKS}
J. So, S. Kwak, and Y. Sung, ``{Some new results on index coding when the number of data is less than the number of receivers}", in \emph{Proc. IEEE International Symposium on Information Theory (ISIT)}, Honululu, USA, 2014, pp. 496--500.


\bibitem{RKMV}
R.K. Bhattaram, M.B. Vaddi, and B.S. Rajan, ``A lifting construction for scalar linear index codes," \emph{arXiv preprint arXiv:1510.08592v1 [cs.IT]}, 29 Oct, 2015.

\bibitem{JS}
J.~B. Ebrahimi, and M.~J. Siavoshani, ``On linear index coding from graph homomorphism perspective,'' in \emph{Proc. Information Theory and Applications Workshop (ITA)}, San Diego, USA, 2016.

\bibitem{PK}
V.~K. Gummadi, A. Choudhary, and P. Krishnan, ``{Index coding : Rank-invariant extensions}", \emph{arXiv preprint arXiv:1704.00687v1 [cs.IT]}, 3 Apr, 2017.

\bibitem{CBSR}
C.~Arunachala and B.~S. Rajan, ``Optimal scalar linear index codes for three
classes of two-sender unicast index coding problem,'' presented in \emph{International Symposium of Information Theory and its Applications (ISITA)}, Singapore, Oct, 2018, also available as \emph{arXiv preprint
	arXiv:1804.03823v1 [cs.IT]}, 11 Apr, 2018.

\bibitem{FYGL}
F. Arbabjolfaei and Y.~-H. Kim, ``{Generalized lexicographic products
and the index coding capacity}", \emph{arXiv preprint  arXiv:1608.03689v2 [cs.IT]}, 30 Sep, 2018.

\bibitem{BY}
Z.~Bar-Yossef, Y.~Birk, T.~S.~Jayram, and T.~Kol, ``Index coding with side-
information,'' \emph{IEEE Transactions on Information Theory}, vol. 57, no. 3, pp. 1479-1494, 2011.


\bibitem{DSC}
S.~H. Dau, V. Skachek, and Y.~M. Chee, ``On the security of index
coding with side information,"
\emph{IEEE Transactions on Information Theory
}, vol. 58, no. 6,
pp. 3975--3988, 2012.

\bibitem{MBSRLC}
M.~B. Vaddi and B.~S. Rajan, ``Low-complexity decoding for symmetric, neighboring and consecutive side-information index coding problems," \emph{arXiv preprint arXiv:1705.03192v2 [cs.IT]}, 16 May, 2017.

\end{thebibliography}
\end{document}